\documentclass[book,numbers]{elsbook-P04}

\usepackage{amsmath,amssymb,amsfonts,amsthm,makeidx,graphicx,colortbl,chapterbib,natbib,algorithm,subfigure}

\def\mSigma{\mbox{$\mathbf{\Sigma} \kern .08em$}}
\def\mLambda{\mbox{$\mathbf{\Lambda} \kern .08em$}}

\DeclareMathOperator{\rank}{\operatorname{rank}}

\DeclareMathOperator{\diag}{\mathrm{diag} \, }

\newcommand{\oline}[1]{\mkern 1.5mu\overline{\mkern-1.5mu#1\mkern-1.5mu}\mkern 1.5mu}

\newcommand{\E}{{\cal E}}
\newcommand{\G}{{\cal G}}

\newcommand{\V}{{\cal V}}

\newcommand{\abs}[1]{\left\vert#1\right\vert}

\def\S{\text{\mbox{${\cal S}$}}}
\def\F{\text{\mbox{${\cal F}$}}}

\def\V{\text{\mbox{${\cal V}$}}}
\def\E{\text{\mbox{${\cal E}$}}}
\def\G{\text{\mbox{${\cal G}$}}}
\def\B{\text{\mbox{${\cal B}$}}}
\def\D{\text{\mbox{${\cal D}$}}}

\def\b0{\text{\mbox{\boldmath $0$}}}
\def\ba{\text{\mbox{\boldmath $a$}}}

\def\bd{\text{\mbox{\boldmath $d$}}}

\def\bp{\text{\mbox{\boldmath $p$}}}

\def\bs{\text{\mbox{\boldmath $s$}}}
\def\bu{\text{\mbox{\boldmath $u$}}}

\def\bx{\text{\mbox{\boldmath $x$}}}
\def\by{\text{\mbox{\boldmath $y$}}}
\def\bz{\text{\mbox{\boldmath $z$}}}
\def\buno{\text{\mbox{\boldmath $1$}}}

\def\bv{\text{\mbox{\boldmath $v$}}}

\def\mA{\mbox{$\mathbf{A}$}}
\def\mB{\mbox{$\mathbf{B}$}}

\def\mC{\mbox{$\mathbf{C}$}}
\def\mD{\mbox{$\mathbf{D}$}}

\def\mI{\mbox{$\mathbf{I}$}}
\def\mL{\mbox{$\mathbf{L}$}}
\def\mP{\mbox{$\mathbf{P}$}}
\def\mQ{\mbox{$\mathbf{Q}$}}
\def\mR{\mbox{$\mathbf{R}$}}
\def\mS{\mbox{$\mathbf{S}$}}

\def\mU{\mbox{$\mathbf{U}$}}

\begin{document}

\chapter[Chapter Title]{Sampling and Recovery \newline of Graph Signals}\label{chap1}

\author*[1]{Paolo Di Lorenzo}%
\author[2]{Sergio Barbarossa,}%
\author[1]{Paolo Banelli}%

\address[1]{\orgname{University of Perugia}, \orgdiv{Department of Engineering}, \orgaddress{Via G. Duranti 93, 06125, Perugia, Italy.}}
\address[2]{\orgname{Sapienza University of Rome}, \orgdiv{Department of Information Engineering, Electronics, and Telecommunications}, \orgaddress{via Eudossiana 18, 00184, Rome, Italy.}}
\address*[3]{Corresponding: \email{paolo.dilorenzo@unipg.it}}

\makechaptertitle

\begin{abstract}[Abstract]
The aim of this chapter is to give an overview of the recent advances related to sampling and recovery of signals defined over graphs. First, we illustrate the conditions for perfect recovery of bandlimited graph signals from samples collected over a selected set of vertexes. Then, we describe some sampling design criteria proposed in the literature to mitigate the effect of noise and model mismatching when performing graph signal recovery. Finally, we illustrate algorithms and optimal sampling strategies for adaptive recovery and tracking of dynamic graph signals, where both sampling set and signal values are allowed to vary with time. Numerical simulations carried out over both synthetic and real data illustrate the potential advantages of graph signal processing methods for sampling, interpolation, and tracking of signals observed over irregular domains such as, e.g., technological or biological networks.
\end{abstract}

\begin{keywords}[Keywords:]
Graph signal processing \sep sampling on graphs \sep interpolation on graphs \sep adaptation and learning over networks
\end{keywords}

\section{Introduction}
\label{intro}

In a large number of applications involving sensor, transportation, communication, social, or biological networks, the observed data can be modeled as signals defined over graphs, or graph signals for short. As a consequence, over the last few years, there was a surge of interest in developing novel analysis methods for graph signals, thus leading to the research field known as graph signal processing (GSP), see, e.g., \cite{shuman2013emerging,sandryhaila2013discrete}. The goal of GSP is to extend classical processing tools to the analysis of signals defined over an irregular discrete domain, represented by a graph, and one interesting aspect is that such methods typically come to depend on the graph topology, see, e.g., \cite{sandryhaila2013discrete,sandryhaila2014discrete,chen2015signal,zhu2012approximating,chen2015discrete}.

A fundamental task in GSP is to infer the values of a graph signal by interpolating the samples collected from a known set of vertices. In the GSP literature, this learning task is known as \textit{interpolation from samples}, and emerges whenever cost constraints limit the number of vertices that we can directly observe. This arises in several applications such as semi-supervised learning of categorical data \cite{gadde2014active}, environmental monitoring \cite{janssen2008spatial}, and missing value prediction as in matrix completion problems \cite{candes2009exact}.
Interpolation methods on graphs rely on the implicit assumption that nodes close to each other have similar values, i.e., the graph encodes similarity
among the values observed over the vertices. For instance, in an item-item graph in a recommendation system, a user would rate two similar items with similar ratings \cite{gomez2016netflix}. In the same way, predicting the functions of proteins based on a protein network relies on some notion of closeness among the nodes \cite{yamanishi2004protein}. In other words, the signals of interest must be smooth functions over the graph. In GSP, the smoothness assumption is typically formalized in terms of (approximate) bandlimitedness over a graph Fourier basis, and enables recovery of the signal after sampling over a selected subset of vertices.

A first seminal contribution to sampling theory in GSP is given by \cite{pesenson2008sampling}, where a sufficient condition for unique recovery is stated for a given sampling set; the approach was then extended in \cite{narang2013signal,anis2014towards}. Most of the works on graph sampling theory assume that a portion of the graph Fourier basis is explicitly known. For example, the work in \cite{chen2015discrete} provides conditions that guarantee unique reconstruction of signals spanned over a subset of vectors composing the graph Fourier basis, proposing also a greedy method to select the sampling set in order to minimize the effect of sample noise in the worst-case. Reference \cite{wang2014local} exploited a smart partitioning of the graph in local-sets, and proposed iterative methods to reconstruct bandlimited graph signal from sampled data. The work in \cite{tsitsvero2015signals} creates a conceptual link between uncertainty principle and sampling of graph signals, and proposes several optimality criteria (e.g., the mean-square error) to select the sampling set in the presence of noise. Another valid approach is the so-called aggregation sampling \cite{marques2016sampling}, which involves successively shifting a signal using the adjacency matrix and aggregating the values at a given node. Greedy sampling strategies with provable performance guarantees were proposed in \cite{chamon2017greedy} in a Bayesian reconstruction setting.

If the size of the graph signal is very large as, e.g., in web-scale graphs \cite{tremblay2016compressive}, complexity becomes a crucial issue, such that in many cases we cannot assume to know or efficiently compute the graph Fourier basis. Some works have then proposed sampling methods that
do not require such previous knowledge. For instance, the work in \cite{anis2016efficient} proposes efficient methods to select the sampling set based on powers of the variation operator to approximate the bandwidth of the graph signal. There are also alternative approaches that do not consider graph spectral information and rely only on vertex-domain characteristic, e.g., maximum graph cuts \cite{narang2010local} and spanning trees \cite{nguyen2015downsampling}. Finally, there exist randomized sampling strategies, e.g., \cite{chen2016signal,puy2016random,Tre17}. The work in \cite{chen2016signal} provides an efficient design of sampling probability distribution over the nodes, deriving bounds on the reconstruction error in the presence of noise and/or approximatively bandlimited signals. Reference \cite{puy2016random} exploits compressive sampling arguments to derive random sampling strategies with variable density, thus also proposing a fast technique to estimate the optimal sampling distribution accurately. Last, the work in \cite{Tre17} proposes a sampling strategy tailored for large-scale data based on random walks on graphs.

The sampling strategies described so far involve batch methods for sampling and recovery of graph signals. In many applications such as, e.g., transportation networks, brain networks, or communication networks, the observed graph signals are typically time-varying. This requires the development of effective methods capable to learn and track dynamic graph signals from a carefully designed, possibly time-varying, sampling set. Some previous works have considered this specific learning task, see, e.g., \cite{dilo2016adaptive,di2016distAdaGraph,romero2016kernel,wang2015distributed}. Specifically, \cite{dilo2016adaptive} proposed an LMS estimation strategy enabling adaptive learning and tracking from a limited number of smartly sampled observations. The LMS method in  \cite{dilo2016adaptive} was then extended to the distributed setting in \cite{di2016distAdaGraph}. The work in \cite{romero2016kernel} proposed a kernel-based reconstruction framework to accommodate time-evolving signals over possibly time-evolving topologies, leveraging spatio-temporal dynamics of the observed data. Finally, reference \cite{wang2015distributed} proposes a distributed method for tracking
bandlimited graph signals, assuming perfect observations and a fixed sampling strategy.

In this chapter, we review some of the recent advances related to sampling and recovery of signals defined over graphs. Due to space limitations, such review will be limited only to some specific contributions.
The structure of the chapter is explained in the sequel. Sec. \ref{Sec:Not_Back} defines the adopted notation, and recalls some background on GSP. In Sec. \ref{Sec:Batch_samp}, we illustrate the conditions for perfect recovery of bandlimited graph signals from samples collected according to design criteria proposed to mitigate the effect of noise or model mismatching. Finally, Sec. \ref{Sec:Adaptive_sampl_rec} illustrates algorithms and optimal sampling strategies for adaptive recovery and tracking of dynamic graph signals, where both sampling set and signal values are allowed to vary with time.

\section{Notation and Background}
\label{Sec:Not_Back}

In this paragraph, we first introduce the notation that we will use throughout the chapter. Then, we briefly recall some basic notions from GSP that will be instrumental for the derivations and arguments of the following sections.\\

\noindent \textbf{Notation.} We  indicate  scalars  by  normal  letters (e.g., $a$); vector variables with bold lowercase letters (e.g., $\ba$) and matrix variables  with  bold  uppercase  letters (e.g., $\mA$). Scalars $a_i$ and $a_{ij}$ correspond to the $i$-th entry of $\ba$ and the $ij$-th entry of $\mA$, respectively. We indicate by $\|\ba\|_2$ and $\|\mA\|_2$ the  $\ell_2$ norm  and  the  spectral  norm  of  the  vector $\ba$ and  matrix $\mA$, respectively. If $\mA$ is rectangular, we denote by $\sigma_i(\mA)$ the $i$-th singular value of $\mA$; if $\mA$ is square, $\lambda_i(\mA)$ represents the $i$-th eigenvalue of $\mA$. The trace operator of matrix $\mA$ is indicated with ${\rm Tr}(\mA)$; ${\rm diag}(\ba)$ is a diagonal matrix having $\ba$ as main diagonal; ${\rm rank}(\mA)$ denotes the rank of matrix $\mA$; $\det(\mA)$ represents the determinant of $\mA$, whereas ${\rm pdet}(\mA)$ is the pseudo-determinant of $\mA$, i.e., the product of all non-zero eigenvalues of $\mA$. The superscript $^H$ denotes the hermitian operator, i.e., the conjugate transposition of a vector or a matrix, whereas $\mA^\dagger$ denotes the pseudo-inverse of matrix $\mA$. $\mathbb{E}\{\cdot\}$ represents the expectation operator. A set of elements is denoted by a calligraphic letter (e.g., $\S$), and $|\S|$ represents the cardinality of set $\S$, i.e., the number of elements of $\S$. The symbols $\cup$, $\cap$, and $\setminus$ denote union, intersection, and difference among sets, respectively. Given a set $\S$, we denote its complement set as $\S^c$, i.e., $\V=\S \cup \S^c$ and $\S \cap S^c=\emptyset$. $\mathbf{1}$ denotes the vector of all ones, whereas $\buno_{\S}$ is the set indicator vector, whose $i$-th entry is equal to one, if  $i \in \S$, or zero otherwise. \\

\noindent \textbf{Background.} We consider a graph $\G = (\V, \E)$ consisting of a set of $N$ nodes $\V = \{1,2,..., N\}$, along with a set of weighted edges $\E=\{a_{ij}\}_{i, j \in \V}$, such that $a_{ij}>0$, if there is a link from node $j$ to node $i$, or $a_{ij}=0$, otherwise. The adjacency matrix $\mA$ of a graph is the collection of all the weights $a_{ij}, i, j = 1, \ldots, N$. The combinatorial Laplacian matrix is defined as $\mathbf{L} = {\rm diag}(\mathbf{1}^T\mA)-\mathbf{A}$. A signal $\bx$ over a graph $\G$ is defined as a mapping from the vertex set to the set of complex numbers, i.e., $\bx: \V \rightarrow \mathbb{C}$. The graph $\G$ is endowed with a graph-shift operator $\mS$ defined as an $N \times N$ matrix whose entry $(i,j)$, denoted with $S_{ij}$, can be non-zero only if $i=j$ or the link $(j,i)\in \E$. The sparsity pattern of matrix $\mS$ captures the local structure of $\G$; common choices for $\mS$ are the adjacency matrix \cite{sandryhaila2013discrete}, the Laplacian \cite{shuman2013emerging}, and its generalizations \cite{anis2016efficient}. We assume that $\mS$ is diagonalizable, i.e., there exists an $N \times N$ matrix $\mU=[\bu_1,\ldots,\bu_N]$ and an $N \times N$ diagonal matrix $\boldsymbol{\Lambda}$ that can be used to decompose $\mS$ as $\mS=\mU\boldsymbol{\Lambda}\mU^{-1}$. When $\mS$ is normal, i.e., when $\mS\mS^H=\mS^H\mS$, matrix $\mU$ is unitary and $\mU^{-1}=\mU^H$.

Recovery of a signal from its sampled version is possible under the assumption that $\bx$ admits a \textit{sparse representation}. The basic idea when addressing the problem of sampling graph signals is to suppose that $\mS$ plays a key role in explaining the signal of interest. More specifically, we assume that $\bx$ can be expressed as a linear combination of a subset of the columns of $\mU$, i.e.,
\begin{equation}
\label{x=Us}
\bx=\mU \bs
\end{equation}
where $\bs\in \mathbb{C}^N$ is either exactly or approximately sparse. In this context, vectors $\{\bu_i\}_{i=1}^N$ are interpreted as the graph Fourier basis, and $\{s_i\}_{i=1}^N$
are the corresponding graph signal frequency coefficients, i.e.,
\begin{equation}\label{GFT}
\bs = \mathbf{U}^H \bx
\end{equation}
takes the role of the Graph Fourier Transform (GFT) of signal $\bx$. As an example, in many cases the graph signal exhibits clustering features, i.e., it is a smooth function over each cluster (e.g., in semi-supervised learning \cite{gadde2014active}), but it may vary arbitrarily from one cluster to the other. In such a case, if the columns of $\mU$ are chosen to represent clusters, the only nonzero (or approximately nonzero) entries of $\bs$ are the ones associated to the clusters. In the case of {\it undirected} graphs, $\mU$ may be composed from the eigenvectors of the Laplacian, which have well known clustering properties \cite{godsil2013algebraic}.

The localization properties of graph signals in vertex and frequency domains will play an important role in the ensuing arguments. To introduce such properties, we first define the matrix $\mU_{\F}\in \mathbb{C}^{N\times|\F|}$, which represents the collection of all the columns of $\mU$ associated with a subset of frequency indices $\F \subseteq \{1,\ldots,N\}$. Then, we introduce the $N\times N$ band-limiting operator
\begin{equation}
\label{lowpass_operator}
\mB_{\F} = \mU_{\F} \mU^H_{\F}.
\end{equation}
The role of $\mB_{\F}$ is to project a vector $\bx$ onto the subspace spanned by the columns of $\mU_{\F}$. Thus, we say that a vector $\bx$ is perfectly localized over the frequency set $\F$ (or $\F$-bandlimited) if
\begin{equation}
\label{Bx=x}
\mB_{\F} \bx=\bx=\mU_{\F}\bs_{\F},
\end{equation}
where $\mB_{\F}$ is given in (\ref{lowpass_operator}), and the second equality comes from (\ref{x=Us}) where we have exploited the sparsity of $\bs$, which is different from zero (and equal to $\bs_{\F}\in \mathbb{C}^{|\F|}$) only in the frequency support $\F$.
Similarly, given a subset of vertices $\S \subseteq \V$, we define the $N\times N$ vertex-limiting operator
\begin{equation}
\label{D}
\mathbf{D}_{\S} = {\rm diag}\{\buno_{\S}\}.
\end{equation}
Thus, we say that a vector $\bx$ is perfectly localized over the subset $\S \subseteq \V$ (or $\S$-vertex-limited) if
$\mD_{\S} \bx=\bx$, with $\mD_{\S}$ defined as in (\ref{D}).
We also denote by $\B_{\F}$ the set of all $\F$-bandlimited signals, and by $\D_{\S}$ the set of all $\S$-vertex-limited signals. The operators $\mathbf{D}_{\S}$ and $\mathbf{B}_{\F}$ are self-adjoint and idempotent, and represent orthogonal projectors onto the sets $\D_{\S}$ and $\B_{\F}$, respectively. Differently from continuous-time signals, a graph signal can be perfectly localized in {\it both} vertex and frequency domains. This property is formally stated in the following theorem \cite[Th. 2.1]{tsitsvero2015signals}.
\begin{theorem}
\label{theorem_unit_eigenvalue}
There is a graph signal $\bx$ perfectly localized over both vertex set $\S$ and frequency set $\F$ (i.e. $\bx \in \B_{\F} \cap \D_{\S}$) if and only if the operator $\mB_{\F}  \mD_{\S}  \mB_{\F}$ has an eigenvalue equal to one; in such a case, $\bx$ is the eigenvector of $\mB_{\F}  \mD_{\S}  \mB_{\F}$ associated to the unit eigenvalue.
\end{theorem}
Perfect localization onto the sets $\S$ and $\F$ can be equivalently expressed in terms of the operator $\mD_{\S} \mU_{\F}$ \cite{tsitsvero2015signals}, i.e., it holds if and only if
\begin{equation}
\label{|BD|=1=|DB|}
\|\mD_{\S} \mU_{\F} \|_2=\|\mB_{\F}\mD_{\S}\mB_{\F}\|_2  = 1.
\end{equation}

In the following section, we will illustrate the theory behind sampling and recovery of signals defined over graphs.

\section{Sampling and Recovery}\label{Sec:Batch_samp}

Let us consider the observation of an $\F$-bandlimited graph signal over the sampling set $\S$. The observation model can be cast as:
\begin{equation}
\label{eq::sampling equation}
\by_{\S}=\mathbf{P}_{\S}^T\, \bx=\mathbf{P}_{\S}^T\mathbf{U}_{\F}\bs_{\F},
\end{equation}
where $\by_{\S}\in \mathbb{C}^{|\S|}$ is the observation vector over the vertex set $\S$, and $\mathbf{P}_{\S}\in \mathbb{R}^{N\times |\S|}$ is a sampling matrix whose columns are indicator functions for nodes in $\S$, and such that the orthogonal projector over $\D_{\S}$ is given by $\mD_{\S}=\mP_{\S}\mP_{\S}^T$ [cf. (\ref{D})]. The problem of recovering a bandlimited graph signal from its samples is then equivalent to the problem of properly selecting the sampling set $\S$, and then recover $\bx$ from $\by_{\S}$ by inverting the system of equations in (\ref{eq::sampling equation}). This approach is known as \textit{selection sampling} and was addressed, for example, in  \cite{pesenson2008sampling}, \cite{narang2013signal}, \cite{chen2015discrete}, and \cite{tsitsvero2015signals}.

In the sequel, we will first consider the conditions for perfect recovery of bandlimited graph signals. Then, we will illustrate the effect of noise and model mismatching on the reconstruction performance. Also, since the identification of the sampling set $\S$ plays a key role in the conditions for signal recovery and in the reconstruction performance, we will illustrate optimization strategies to design the sampling set. Finally, we will illustrates results of numerical simulations carried out over synthetic and realistic data.

\subsection{Sampling and Perfect Recovery of bandlimited graph signals}
\label{Perf_rec}

We will now address the fundamental problem of assessing the conditions and the means for perfect recovery of $\bx$ from $\by_{\S}$. To this aim, we introduce the operator $\mD_{\S^c}=\mI-\mD_{\S}$, which projects onto the complement vertex set $\S^c = \V \setminus \S$. Starting from (\ref{eq::sampling equation}), the necessary and sufficient conditions for perfect recovery are stated in the following Theorem \cite[Th. 4.1]{tsitsvero2015signals}.
\begin{theorem}
\label{theorem::sampling theorem}
Any $\F$-bandlimited graph signal $\bx$ can be perfectly recovered from its samples collected over the vertex set $\S$ if and only if
\begin{equation}
\label{|DcB|<1}
\| \mD_{\S^c} \mU_{\F} \|_2 < 1,
\end{equation}
i.e., if there are no $\F$-bandlimited signals that are perfectly localized on $\S^c$.
\end{theorem}

\begin{proof}
From (\ref{eq::sampling equation}), a sufficient condition for signal recovery is the existence of the pseudo-inverse matrix $\mQ=(\mP_{\S}^T\mU_{\F})^\dagger=(\mU_{\F}^H\mD_{\S}\mU_{\F})^{-1}\mU_{\F}^H\mP_{\S}$.
Since  $\mU_{\F}^H\mD_{\S}\mU_{\F}=\mI-\mU_{\F}^H\mD_{\S^c} \mU_{\F}$, we obtain that $\mQ$ exists if
$\|\mU_{\F}^H\mD_{\S^c} \mU_{\F}\|_2=\|\mD_{\S^c} \mU_{\F}\|_2<1$, i.e., if (\ref{|DcB|<1}) holds true. Conversely, if $\|\mD_{\S^c}\mU_{\F}\|_2=1$, there exist bandlimited signals that are perfectly localized over $\S^c$ [cf. (\ref{|BD|=1=|DB|})]. Thus, if we sample one of such signals over $\S$, it would be impossible to recover $\bx$ from those samples. This proves that condition (\ref{|DcB|<1}) is also necessary.
\end{proof}




\noindent Theorem \ref{theorem::sampling theorem} and its proof also suggest the reconstruction formula:
\begin{equation}
\label{eq::rec2}
\widehat{\bx}=\mU_{\F}(\mP_{\S}^T\mU_{\F})^\dagger\by_{\S}=\mU_{\F}(\mU_{\F}^H\mD_{\S}\mU_{\F})^{-1}\mU_{\F}^H\mP_{\S}\by_{\S},
\end{equation}
which guarantees reconstruction of the bandlimited graph signal $\bx$ if condition (\ref{|DcB|<1}) holds true, and has computational complexity equal to $O(N|\F|^2)$. The above reconstruction formula is also known as consistent reconstruction \cite{eldar2003sampling} since it keeps the observed samples unchanged.

Let us consider now the implications of condition (\ref{|DcB|<1}) of Theorem \ref{theorem::sampling theorem} on the sampling strategy. To fulfill (\ref{|DcB|<1}), we need to guarantee that there exist no signals that are perfectly localized over the vertex set $\S^c$ and the frequency set $\F$. Since, in general, we have
\begin{equation}
\by_{\S}= \mP_{\S}^T\mU_{\F} \bs_{\F}+\mP_{\S^c}^T \mU_{\F} \bs_{\F},
\end{equation}
we need to guarantee that $\mP_{\S}\mU_{\F} \bs_{\F} \neq\b0$ for any non-trivial vector $\bs_{\F}$,
which requires $\mP_{\S}^T\mU_{\F}$ to be full column rank, i.e.,
\begin{equation}
\label{db_equal_b}
\rank (\mP_{\S}^T \mU_{\F}) = \rank (\mU_{\F}) = |\F|.
\end{equation}
Of course, a necessary condition to satisfy (\ref{db_equal_b}) is that
\begin{equation}\label{S>=F}
|\S| \ge |\F|.
\end{equation}
However, condition (\ref{S>=F}) is not sufficient, because $\mP_{\S}^T\mU_{\F}$ may loose rank, depending on graph topology and samples' location. As a particular case, if the graph is not connected, the vertices can be labeled so that the Laplacian (adjacency) matrix can be written as a block diagonal matrix, with a number of blocks equal to the number of connected components. Correspondingly, each eigenvector of $\mL$ can be expressed as a vector having all zero elements, except for the entries corresponding to the connected component. This implies that, if there are no samples over the vertices corresponding to the non-null entries of the eigenvectors with index included in $\F$, $\mP_{\S}^T\mU_{\F}$ looses rank.
More generally, even if the graph is connected, there may easily occur situations where matrix $\mP_{\S}^T\mU_{\F}$ is not rank-deficient, but it is ill-conditioned, depending on graph topology and samples' location. This case is particularly dangerous when the true signal is only approximately bandlimited (which is the case for most signals in practice) or when the samples are noisy. In such cases, not all sampling sets of given size are equally good, and it becomes fundamental to understand which is the best sampling set that achieves the smallest reconstruction error. Thus, in the sequel, we will first illustrate the effect of noise and model mismatching on graph signal reconstruction, and then we will describe sampling strategies satisfying several optimization criteria.

\subsection{The effect of noise and model mismatching}
\label{Sec:noise_aliasing}
Let us consider first the reconstruction of bandlimited signals from noisy samples, where the observation model is given by:
\begin{equation}
\label{r=D(s+n)}
\by_{\S} = \mP_{\S}^T \left( \bx + \bv \right) = \mP_{\S}^T \mU_{\F}\bs_{\F} + \mP_{\S}^T\bv,
\end{equation}
where $\bv$ is a zero-mean noise vector with covariance matrix $\mR_v=\mathbb{E}\{\bv\bv^H\}$. To design an interpolator in the presence of noise, we consider the best linear unbiased estimator (BLUE), which is given by \cite{kay1993fundamentals}:
\begin{equation}
\label{eq::rec3}
\widehat{\bx}= \mU_{\F} \Big(\mU_{\F}^H\mP_{\S} \left(\mP_{\S}^T\mR_v\mP_{\S}\right)^{-1}\mP_{\S}^T\mU_{\F}\Big)^{-1}\mU_{\F}^H\mP_{\S}\left(\mP_{\S}^T\mR_v\mP_{\S}\right)^{-1} \by_{\S}.
\end{equation}
The estimator in (\ref{eq::rec3}) minimizes the least square error and, if noise is Gaussian in (\ref{r=D(s+n)}), it coincides with the minimum variance unbiased estimator, which attains the Cram\'er-Rao lower bound. It is immediate to see that (\ref{eq::rec3}) is an unbiased estimator, i.e., $\mathbb{E}\{\widehat{\bx}\}=\bx$. Furthermore, the mean square error (MSE) is given by \cite{kay1993fundamentals}:
\begin{align}
\label{eq::MSE_batch}
{\rm MSE}=\mathbb{E}\|\widehat{\bx}-\bx\|^2={\rm Tr}\left\{ \Big(\mU_{\F}^H\mP_{\S} \left(\mP_{\S}^T\mR_v\mP_{\S}\right)^{-1}\mP_{\S}^T\mU_{\F}\Big)^{-1}\right\}.
\end{align}
As a particular case, if noise is spatially uncorrelated, i.e., $\mR_v={\rm diag}\{r_1^2,\ldots,r_N^2\}$, and letting $\bu_{\F,i}^H$ be the $i$-th row of matrix $\mU_{\F}$, we obtain:
\begin{align}
\label{eq::MSE_batch_unc}
{\rm MSE}={\rm Tr}\left\{ \left(\mU_{\F}^H\mD_{\S}\mR_v^{-1}\mU_{\F}\right)^{-1} \right\}={\rm Tr}\left\{ \left(\sum_{i\in \S} \bu_{\F,i}\bu_{\F,i}^H/r_i^2\right)^{-1} \right\}.
\end{align}
This illustrates how, in the presence of uncorrelated noise, the design of the sampling set should minimize the trace of the inverse of matrix $\mU_{\F}^H\mD_{\S}\mR_v\mU_{\F}$.


So far we assumed that the true signal $\bx$ is perfectly bandlimited, i.e., $\bx\in \B_{\F}$. However, in most applications, the signals are only approximately bandlimited. In such a case, the recovery formula in (\ref{eq::rec2}) applied to such signals leads to a reconstruction error, which is analyzed next. In general, an approximately bandlimited graph signal can be expressed as
\begin{equation}\label{mismatch}
\bx=\bx_{\F}+\Delta\bx,
\end{equation}
where $\bx_{\F}=\mB_{\F}\bx$ is the bandlimited component, whereas $\Delta\bx=\mB_{\F^c} \bx$ represents the model mismatch. Sampling the signal over the vertex set $\S$ and using (\ref{eq::rec2}) as a recovery formula, then an upper bound (i.e., a worst-case) on the reconstruction error is given by \cite{eldar2003sampling}:
\begin{equation}\label{upper_bound_aliasing}
\|\widehat{\bx}-\bx\|\leq \frac{\|\Delta \bx\|}{\cos(\theta_{\max})},
\end{equation}
where $\theta_{\max}$ represents the maximum angle between the subspaces $\B_{\F}$ and $\D_{\S}$, which is defined as:
\begin{equation}
\label{cos_theta}
\begin{aligned}
\cos(\theta_{\max})=& \;\;\;\underset{\|\bz\|=1}{\inf} \;\; \| \mathbf{D}_{\S}  \bz\|_2\\
 &\;\; \text{subject to}\;\; \mathbf{B}_{\F} \bz=\bz.
\end{aligned}
\end{equation}
In particular, from (\ref{cos_theta}), it is easy to see that $\cos(\theta_{\max})>0$ if condition (\ref{|DcB|<1}) holds true. Intuitively, the bound in (\ref{upper_bound_aliasing}) says that, for the worst-case error to be minimum, the sampling and reconstruction subspaces should be as aligned as possible. Therefore, for approximatively bandlimited signals, an optimal sampling set should be selected in order to maximize the smallest maximum angle between the subspaces $\B_{\F}$ and $\D_{\S}$. Interestingly, from (\ref{cos_theta}) and (\ref{lowpass_operator}), it appears clear that
\begin{equation}\label{cos_theta2}
\cos(\theta_{\max})=\sigma_{\min}(\mD_{\S}\mU_{\F})
\end{equation}
Thus, in the presence of model mismatching, the design of the sampling set should maximize the minimum singular value of matrix $\mD_{\S}\mU_{\F}$ or, equivalently, the minimum eigenvalue of matrix $\mU_{\F}^H\mD_{\S}\mU_{\F}$.

In the next section, we will illustrate the strategies used to optimize the selection of the sampling set.

\subsection{Sampling strategies}
\label{Sec:sampl_strat}

As previously mentioned, when sampling graph signals, besides choosing the right number of samples, whenever possible it is also fundamental to have a strategy indicating {\it where} to sample, as the samples' location plays a key role in the performance of reconstruction algorithms. In principle, in the ideal case (\ref{eq::sampling equation}), any sampling set $\S$ that satisfies condition (\ref{|DcB|<1}) enables unique reconstruction through the interpolation formula in (\ref{eq::rec2}). However, in the presence of noise or model mismatching, from (\ref{eq::MSE_batch_unc}) and (\ref{upper_bound_aliasing})-(\ref{cos_theta2}), it is clear that the quality of reconstruction is strongly affected by a careful design of the sampling set $\S$.  Different costs can then be defined to measure the reconstruction error and are based on optimal design of experiments \cite{winer1971statistical}. For instance, if we seek for the optimal sampling set $\S^{\rm opt}$ of size $M$, as the set that minimizes the mean squared error in (\ref{eq::MSE_batch_unc}), we have:
\begin{equation}\label{A-design}
\displaystyle \S^{\rm A-opt}=\underset{|\S|=M}{\arg\min}\;\;{\rm Tr}\left\{ \left(\mU_{\F}^H\mD_{\S}\mR_v^{-1}\mU_{\F}\right)^{-1} \right\}.
\end{equation}
This is analogous to the so-called \textit{A-optimal design} \cite{winer1971statistical}, and is equivalent to the one proposed in \cite{tsitsvero2015signals}. Similarly, if we aim to design the optimal sampling set of size $M$ to minimize the worst-case reconstruction error in the presence of model mismatching [cf. (\ref{upper_bound_aliasing})-(\ref{cos_theta2})], we have:
\begin{equation}\label{E-design}
\displaystyle \S^{\rm E-opt}=\underset{|\S|=M}{\arg\max}\;\;\sigma_{\min}(\mD_{\S}\mU_{\F}),
\end{equation}
which is equivalent to the so-called \textit{E-optimal design} \cite{winer1971statistical}. The above criterion is equivalent to the one proposed in \cite{chen2015discrete} and, in general, it is useful to find a stable sampling set that satisfies condition (\ref{|DcB|<1}). To select the optimal sampling set, we should solve one of the problems in (\ref{A-design}) or (\ref{E-design}), which entail the selection of an $M$-element subset of $\mathcal{V}$ that optimizes the adopted design criterion. This is a finite combinatorial optimization problem (which is known to be NP-hard \cite{nemhauser1988integer}), whose solution in general requires an exhaustive search over all the possible combinations. Since the number of possible subsets grows factorially as $|\mathcal{V}|$ increases, a brute force approach quickly becomes infeasible also for graph signals of moderate dimensions. To cope with this issue, in the sequel, we will introduce lower complexity methods based on: i) greedy approaches, and ii) convex relaxations.

\subsubsection{Greedy Sampling}
In this section, we  will consider a numerically efficient, albeit sub-optimal, greedy algorithm to tackle the problem of selecting the sampling set. The greedy approach is described in Algorithm 1.
\begin{algorithm}[t]
$\textit{Input Data}:$ $\mU_{\F}$, $M$;

$\textit{Output Data}:$ $\S$, the sampling set. \smallskip

$\textit{Function}:$ \hspace{.23cm} initialize $\S\equiv \emptyset$

\hspace{2 cm} while $|\S|<M$

\hspace{2.3cm} $\displaystyle s=\arg \max_j \;f(\S\cup\{j\})$;

\hspace{2.3cm} $\S \leftarrow \S \cup \{s\}$;

\hspace{2cm} end

\protect\caption{\label{alg:Greedy1}\textbf{: Greedy selection of graph samples}}
\end{algorithm}
The simple idea underlying such method is to iteratively add to the sampling set those vertices of the graph that lead to the largest increment of an adopted performance metric, i.e., a specific set function $f(\S):2^{\mathcal{V}}\rightarrow \mathbb{R}$. We will set $f(\S)=-{\rm Tr}\left\{(\mU_{\F}^H\mD_{\S}\mR_v^{-1}\mU_{\F})^{-1} \right\}$ if we use an A-optimality design as in (\ref{A-design}), or $f(\S)=\sigma_{\min}(\mD_{\S}\mU_{\F})$ if we consider an E-optimality design as in (\ref{E-design}). In fact, since Algorithm 1 starts from the empty set, when $|\S|<|\F|$, matrix $\mU_{\F}^H\mD_{\S}\mR_v^{-1}\mU_{\F}$ is inevitably rank deficient, and its inverse does not exist. In this case, considering an A-optimality criterion, we can use $f(\S)=-{\rm Tr}\left\{(\mU_{\F}^H\mD_{\S}\mR_v^{-1}\mU_{\F})^\dagger \right\}$, which becomes equivalent to (\ref{A-design}) when condition (\ref{|DcB|<1}) is satisfied.

In general, the performance of the greedy strategy will be sub-optimal with respect to an exhaustive search procedure. Nevertheless, if the set function $f(\S)$ satisfies some structural properties, the greedy Algorithm 1 can  be proved to be close to optimality. In  particular, submodularity plays a similar role in combinatorial optimization to convexity in continuous optimization and shares other features of concave functions \cite{lovasz1983submodular}. \\

\noindent \textbf{Definition 1:} A set function
$f:2^{\mathcal{V}}\rightarrow \mathbb{R}$ is submodular if and only if the derived set functions
$f_a:2^{\mathcal{V}\setminus\{a\}}\rightarrow \mathbb{R}$
\begin{equation}
f_a(\S)=f(\S\cup \{a\})-f(\S)
\end{equation}
are monotone decreasing, i.e., if for all subsets $a,\mathcal{A},\mathcal{B}\subseteq \mathcal{V}$ it holds that
\begin{equation}
\hspace{3.7cm}\mbox{if}\;\;\mathcal{A}\subseteq\mathcal{B}\; \Rightarrow\; f_a(\mathcal{A})\geq f_a(\mathcal{B}). \nonumber\hspace{3cm}\qed\\
\end{equation}

\noindent Intuitively, submodularity is a diminishing returns property where  adding  an  element to a  smaller  set  gives  a  larger  gain than  adding  one  to  a  larger set. The maximization of monotone increasing submodular functions is still NP-hard, but the greedy heuristic can be used to obtain a solution that is provably close to optimality, with a solution having objective value within $1-1/e$ of the optimal combinatorial solution \cite{nemhauser1978analysis}.

Unfortunately, both set functions in (\ref{A-design}) and (\ref{E-design}) are not submodular functions \footnote{Interestingly, in a Bayesian recovery setting \cite{chamon2017greedy}, the negative of the MSE function was proved to be approximatively submodular.} \cite{summerscorrection}. Thus, even if the design criteria in (\ref{A-design}) and (\ref{E-design}) are useful to minimize the effect of noise [cf. (\ref{eq::MSE_batch})] and model mismatching [cf. (\ref{upper_bound_aliasing})-(\ref{cos_theta2})], respectively, we do not have theoretical performance guarantees when applying Algorithm 1 to solve such problems. Nevertheless, in the literature of experimental design, a further design criterion is often considered as a surrogate for (\ref{A-design}) [or (\ref{E-design})], which writes as:
\begin{align}\label{D-design}
\displaystyle \S^{\rm D-opt}&\,=\,\underset{|\S|=M}{\arg\max}\;\;\log \det\left(\mU_{\F}^H\mD_{\S}\mR_v^{-1}\mU_{\F}\right)\nonumber\\
&\,=\,\underset{|\S|=M}{\arg\max}\;\;\log \det\left(\sum_{i\in \S} \bu_{\F,i}\bu_{\F,i}^H/r_i^2\right).
\end{align}
This is analogous to the so-called \textit{D-optimal design} \cite{winer1971statistical}, and is equivalent to one of the methods proposed in \cite{tsitsvero2015signals} for graph signals sampling. This design strategy aims at maximizing the volume of the parallelepiped built with the selected rows $\{\bu_{\F,i}^H\}_{i\in \S}$ of matrix $\mU_{\F}$ (weighted by the inverse of the noise variances $\{r_i^2\}_{i\in \S}$), and the rationale is to design a well suited basis for the graph signal that we want to estimate. Interestingly, the set function $f(\S)=\log \det(\mU_{\F}^H\mD_{\S}\mR_v^{-1}\mU_{\F})$ is a monotone increasing submodular function \cite{shamaiah2010greedy}. Thus, in this case, the greedy approach in Algorithm 1 can be used to solve (\ref{D-design}) with provable performance guarantees. In the implementation of Algorithm 1, when $|\S|<|\F|$ and matrix $\mU_{\F}^H\mD_{\S}\mR_v^{-1}\mU_{\F}$ is rank-deficient, we can use $f(\S)=\log {\rm pdet}(\mU_{\F}^H\mD_{\S}\mR_v^{-1}\mU_{\F})$, which is equivalent to (\ref{D-design}) when sampling condition (\ref{|DcB|<1}) on perfect recovery is satisfied.

\subsubsection{Convex Relaxation}

Another possible algorithmic solution to problems like \eqref{A-design}, \eqref{E-design}, \eqref{D-design}, is to resort to convex relaxation techniques, see, e.g., \cite{joshi2009sensor,chepuri2015sparsity,chepuri2016sparse}. To this aim, let us introduce the indicator vector $\bd=\{d_i\}_{i=1}^N$, such that the i-th entry is binary and given by $d_i=1$ if node $i$ belongs to the sampling set $\S$, and $d_i=0$ otherwise. Using the indicator vector $\bd$, we can build a general sampling design problem that can be cast as:
\begin{equation}\label{Conv_relax}
\setlength{\jot}{6pt}
\begin{aligned}
&\hspace{-.3cm}\min_{\bd} \;\; f(\bd) \\
&\hbox{s.t.} \;\;\;\,\mathbf{1}^T\bd=M, \\
&\;\;\;\;\;\;\;\;\bd\in\{0,1\}^N,
\end{aligned}
\end{equation}
where $f(\bd)={\rm Tr}\left\{(\mU_{\F}^H{\rm diag}(\bd)\mR_v^{-1}\mU_{\F})^{-1} \right\}$ for the A-optimal design [cf. \eqref{A-design}], $f(\bd)=-\sigma_{\min}({\rm diag}(\bd)\mU_{\F})$ for the E-optimal design [cf. \eqref{E-design}], and $f(\bd)=-\log\det(\mU_{\F}^H{\rm diag}(\bd)\mR_v^{-1}\mU_{\F})$ for the D-optimal design [cf. \eqref{D-design}]. Problem (\ref{Conv_relax}) has still combinatorial complexity, due to the integer nature of the optimization variable $\bd$. Nevertheless, we can simple relax the indicator variable $\bd$ to be a real vector belonging to the hypercube $[0,1]^N$, thus leading to the following formulation:
\begin{equation}\label{Conv_relax2}
\setlength{\jot}{6pt}
\begin{aligned}
&\hspace{-.3cm}\underset{\bd\in[0,1]^N}{\min}\;\; f(\bd) \\
&\hbox{s.t.} \;\;\;\;\;\,\mathbf{1}^T\bd=M.
\end{aligned}
\end{equation}
It is now easy to check that problem (\ref{Conv_relax2}) is convex for all objective functions $f(\bd)$ defined by the design criteria in
\eqref{A-design}, \eqref{E-design}, \eqref{D-design}, and its global solution can be found using efficient numerical methods \cite{boyd2004convex}. Of course, since (\ref{Conv_relax2}) is a relaxed version of (\ref{Conv_relax}), its real solution $\bd^*$ might need a further selection/thresholding step in order to generate a valid integer vector, as required by (\ref{Conv_relax}). For instance, a possible solution is to select the $M$ sampling nodes as the ones associated with the $M$ largest entries of $\bd^*$. Finally, one can also formulate the sampling design problem in the opposite way with respect to (\ref{Conv_relax}). In particular, we might be interested in searching for the optimal indicator vector $\bd$ that minimizes the number of collected samples, i.e., the $\ell_0$ norm of the vector $\bd$, under a performance requirement on the function $f(\bd)$, e.g, the MSE in (\ref{eq::MSE_batch_unc}). This category of design problems takes the name of sparse sensing \cite{chepuri2015sparsity,chepuri2016sparse} and, using similar relaxation arguments as before, such criteria lead to convex optimization problems.

In the next section, we will illustrate some numerical results aimed at assessing the performance of the described sampling and recovery strategies.

\begin{figure}
 \centering
 \subfigure[Graph topology and sampling set]
   {\includegraphics[width=6cm]{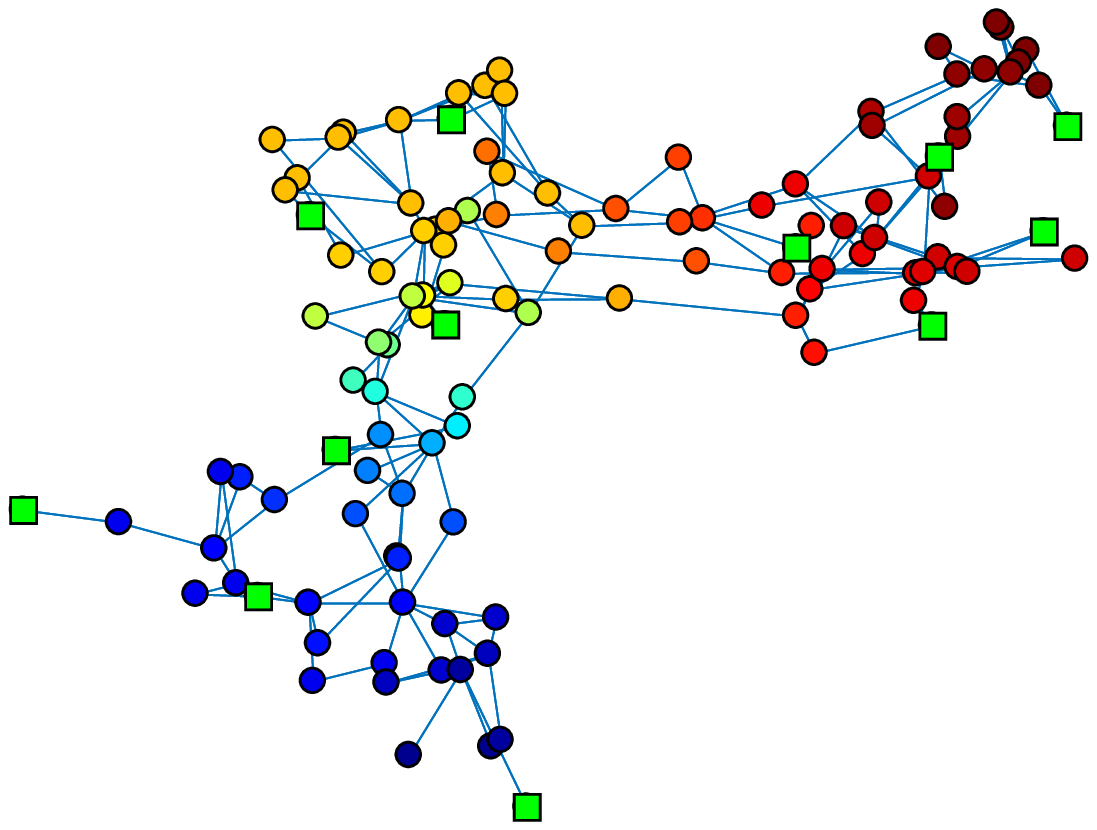}}
    \hspace{0.1cm}
 \subfigure[MSE versus number of samples]
   {\includegraphics[width=5.5cm]{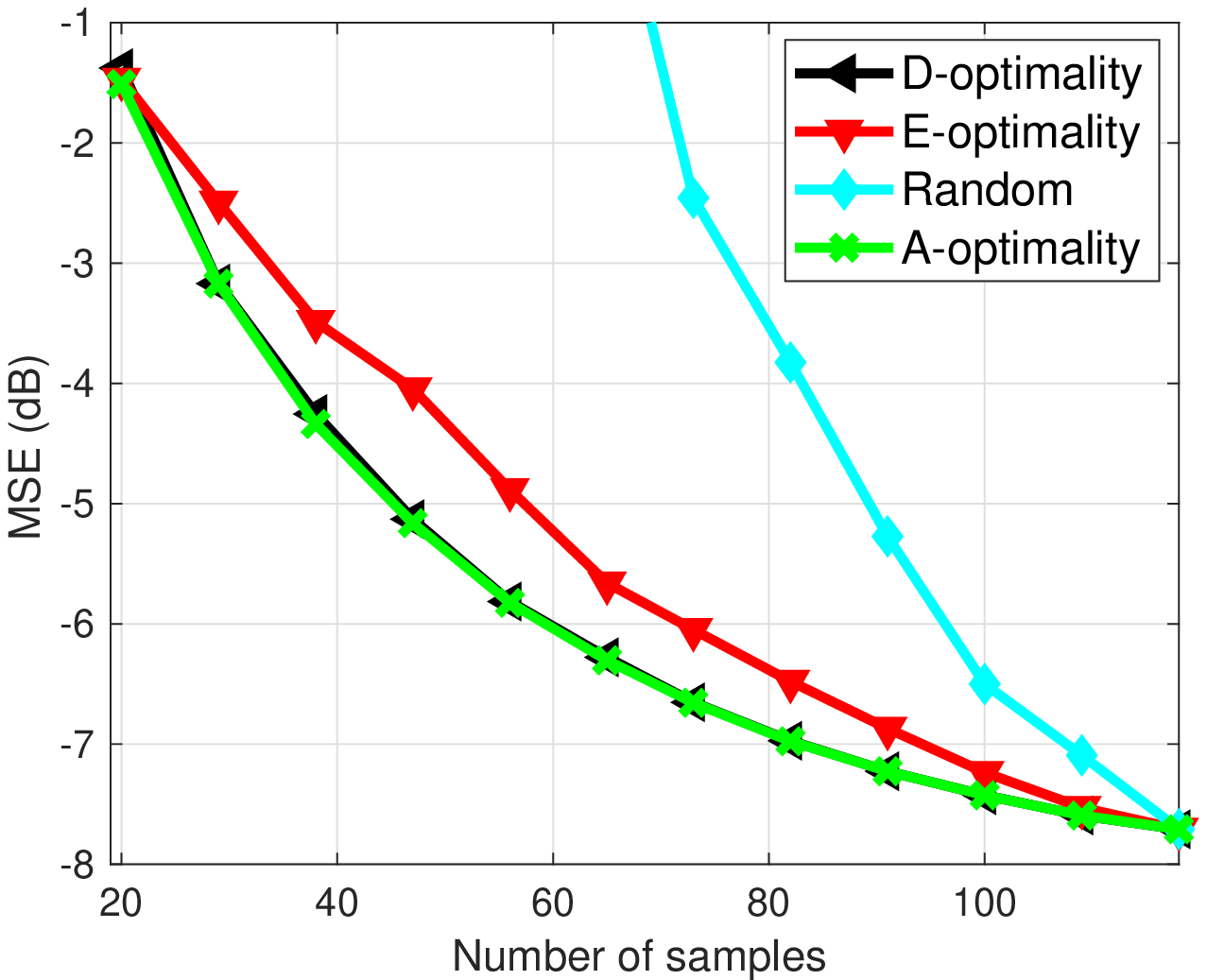}}
 \caption{Sampling and recovery over the IEEE 118 Bus graph}
 \label{Nets}
 \end{figure}

\subsection{Numerical Results}

In the sequel, we consider the application of the described sampling and recovery methods to two real graphs: a power network, and a road network.\\

\noindent \textbf{Sampling over power grids.} The first example involves the IEEE  118  Bus  Test  Case,  i.e.,  a  portion  of  the American Electric Power System (in the Mid-western US) as of December 1962. The graph is composed of 118 nodes (i.e., buses), its  topology  (i.e.,  transmission  lines  connecting buses) is depicted in Fig. \ref{Nets}(a) \cite{sun2005simulation}, and the color of each node encodes the entries of the eigenvector of the Laplacian matrix associated to the second smallest eigenvalue (these entries highlight the presence of three distinct clusters in the network). As illustrated in \cite{pasqualetti2014controllability}, the dynamics of the power generators give rise to smooth graph signals, so that the bandlimited assumption is justified, although in approximate sense. In our example, we randomly generate a lowpass signal with $|\F|=12$ and we take a number of samples equal to $|\S|=12$. The green squares correspond to the samples selected using the greedy Algorithm 1 and the A-optimality design in (\ref{A-design}). It is interesting to see how the method distributes samples over the clusters, and puts the samples, within each cluster, quite far apart from each other. Finally, we compare the reconstruction performance obtained by the considered greedy sampling strategies [cf. (\ref{A-design}), (\ref{E-design}), and (\ref{D-design})] and by random sampling. To this aim we consider graph signal recovery, in the presence of an uncorrelated, zero mean Gaussian random noise with unit variance, and considering $|\F|=20$. Thus, Fig. \ref{Nets}(b) reports the MSE in (\ref{eq::MSE_batch_unc}) versus the number of samples collected over the graph. As expected, the MSE decreases as the number of  samples increases. We can also notice how random sampling performs quite poorly, whereas the A-optimal design (and the D-optimal design) outperforms all other strategies.\\

\begin{figure}[t]
 \centering
 \subfigure[Graph topology and sampling set]
   {\includegraphics[width=5.9cm]{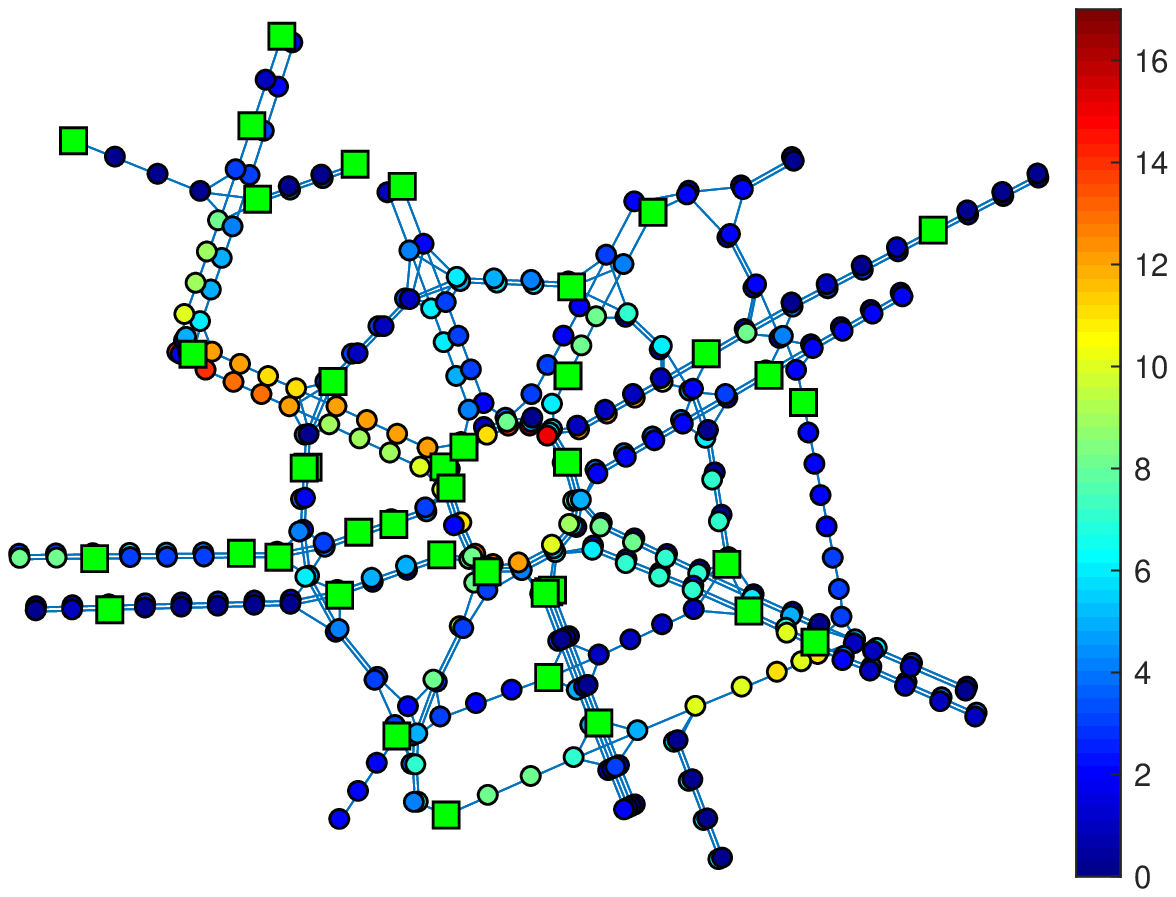}}
 \hspace{0.1cm}
\subfigure[NMSE versus bandwidth]
  { \includegraphics[width=5.5cm]{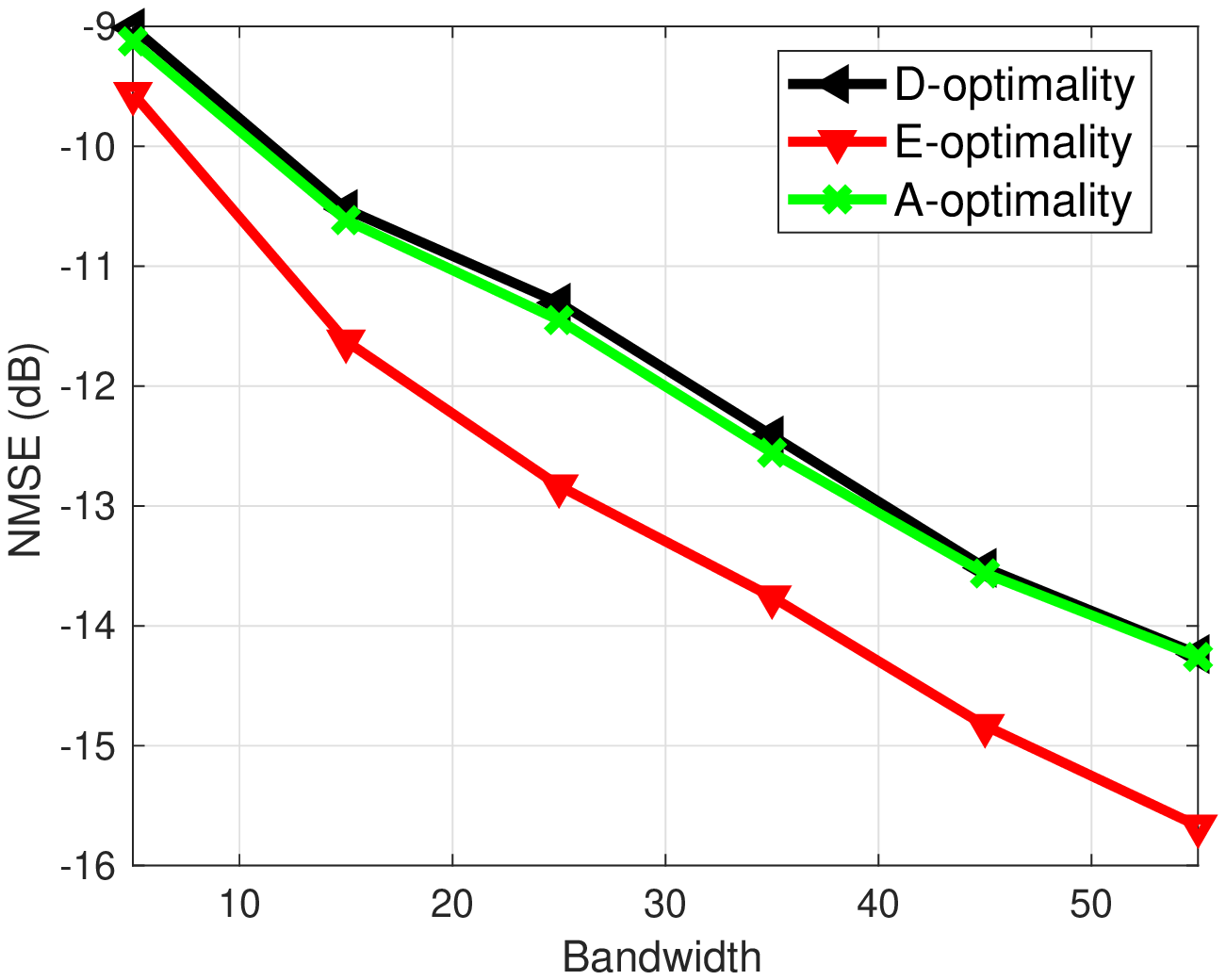}}
 \caption{Sampling and recovery of vehicular flows over road networks}\label{Nets2}
\end{figure}

\noindent \textbf{Traffic flow prediction over road networks.} The second example considers sampling of a portion of the road network in the neighborhood of Mazzini square, which is in the city of Rome, Italy. We have placed landmarks (nodes of the graph) over the streets in a regular fashion, and connected adjacent landmarks on the same lane and at the junctions, thus obtaining the graph topology depicted in Fig. \ref{Nets2}(a). The signal lying on the vertices of the graph represents the flow (number of vehicles per unit of time) of cars passing through the landmarks during a period of 30 seconds, and was obtained using a realistic simulator of urban mobility, namely, SUMO \cite{behrisch2011sumo}. The similarity of values of the signal over adjacent nodes makes the signal to be smooth, but only approximatively bandlimited. In this sense, there is a mismathcing between the observed signal and the bandlimited model used for processing. The goal is to infer the traffic situation over all the road network from a small number of collected samples. Thus, we consider a bandwidth equal to $|\F|=30$, and we take a number of samples equal to $|\S|=40$. The green squares in Fig. \ref{Nets2}(a) correspond to the samples selected using the convex relaxation in (\ref{Conv_relax2}) and the E-optimality design in (\ref{E-design}). It is interesting to see how the method distributes almost uniformly the samples over the streets and the junctions. Finally, we compare the reconstruction performance obtained by the sampling strategies based on convex relaxation [cf. (\ref{A-design}), (\ref{E-design}), and (\ref{D-design})] in the presence of model mismatching. To this aim, Fig. \ref{Nets2}(b) reports the normalized MSE (NMSE), i.e., ${\rm NMSE}=\|\widehat{\bx}-\bx\|^2/\|\bx\|^2$, versus the graph signal bandwidth, and selecting $|\S|=|\F|$. From Fig. \ref{Nets}(b), as expected, the NMSE decreases if we use a larger bandwidth. Furthermore, in this case, the E-optimal design outperforms all other strategies, as expected from (\ref{upper_bound_aliasing})-(\ref{cos_theta2}).

 \subsection{$\ell_1$-norm reconstruction of graph signals}

Let us consider now a different observation model, where a bandlimited graph signal $\bx\in \B_{\F}$ is observed everywhere, but a subset of nodes
$\S$ is strongly corrupted by noise, i.e.,
\begin{equation}
\label{r=D(s+n)}
\by = \bx+\mD_{\S}\bv,
\end{equation}
where the noise is arbitrary but bounded, i.e., $\|\bv\|_1<\infty$. This model is relevant, for example, in sensor networks, where a subset of sensors can be damaged
or  highly  interfered.  The  problem  in  this  case  is  whether  it is possible to recover the graph signal $\bx$ exactly, i.e., irrespective of noise. Even though this is not a sampling problem, the solution is  still  related  to  sampling  theory.  Clearly,  if  the  signal $\bx$
is bandlimited  and  if  the  indexes  of  the  noisy  observations  are known, the answer is simple: $\bx\in \B_{\F}$
can be perfectly recovered from the noisy-free observations, i.e., by completely discarding the noisy observations, if the sampling theorem condition (\ref{|DcB|<1})
holds true. But of course, the challenging situation occurs when the location of the noisy observations is not known. In such a
case, we may resort to an $\ell_1$-norm minimization, by formulating the problem as follows \cite{tsitsvero2015signals}:
\begin{equation}
\label{eq::l1_minimization}
\widehat{\bx} = \arg \min_{\bx \in \B} \| \by-\bx \|_1.
\end{equation}
We provide next some theoretical bounds on the cardinality of $\S$ and $\F$ enabling perfect recovery of the bandlimited graph signal using (\ref{eq::l1_minimization}). To this purpose, we recall the following lemma from \cite{tsitsvero2015signals}.
\begin{lemma}
Let us define
$\mu := \max_{\substack{j \in \F \\ i \in \V}} \abs{u_j (i)}$,
where $u_j (i)$ is the $i$-th entry of the $j$-th vector of the graph Fourier basis. If for some {\it unknown} $\S$, we have
\begin{equation}
\label{eq::unknown_S_condition}
\abs{\S} < \frac{1}{2 \mu^2\abs{\F}},
\end{equation}
then the $\ell_1$-norm reconstruction method (\ref{eq::l1_minimization}) recovers $\bx \in \B$ perfectly, i.e. $\widehat{\bx}=\bx$, for any arbitrary noise $\bv$ present on at most $\abs{\S}$ vertices.
\end{lemma}

\begin{figure}[t]
\centering
\includegraphics[width=6cm,height=4.5cm]{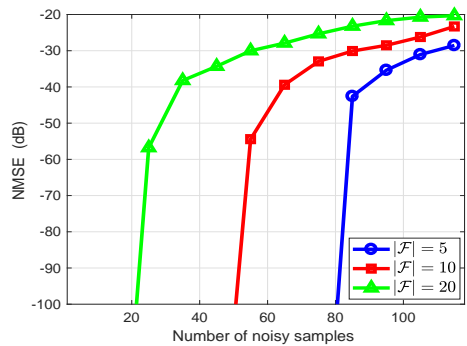}
\caption{$\ell_1$-norm reconstruction: NMSE versus number of noisy samples.}
\label{fig:logan}
\end{figure}

An example of $\ell_1$ reconstruction based on (\ref{eq::l1_minimization}) is useful to
grasp some interesting features. We consider the IEEE 118 bus graph in Fig. \ref{Nets}(a). The signal
is assumed to be bandlimited, with a spectral content limited to the first
$|\F|$ eigenvectors of the Laplacian matrix. In Fig. \ref{fig:logan}, we report the behavior of the NMSE associated to the $\ell_1$-norm
estimate in (\ref{eq::l1_minimization}) , versus the number of noisy samples, considering different values of bandwidth
$|\F|$. As we can notice from
Fig. \ref{fig:logan}, for any value of $|\F|$, there exists a threshold value such that, if the number of noisy samples is lower than the threshold, the reconstruction of the signal is error free. As expected, a smaller signal bandwidth allows perfect reconstruction with a larger number of noisy samples.

\section{Adaptive Sampling and Recovery}
\label{Sec:Adaptive_sampl_rec}

In this section, we consider processing methods capable to learn and track dynamic graph signals from a carefully designed, possibly time-varying, sampling set. To this aim, let us assume that, at each time $n$, noisy samples of the signal are taken over a (randomly) time-varying subset of vertices, according to the following model:
\begin{align}
\label{lin_observation}
\by[n]\,=\,&\mD_{\S[n]}\left(\bx+\bv[n]\right)=\,\mD_{\S[n]}\mU_{\F}\bs_{\F}+\mD_{\S[n]}\bv[n]
\end{align}
where $\mD_{\S[n]}={\rm diag}\{d_i[n]\}_{i=1}^N\in\mathbb{R}^{N\times N}$ [cf. (\ref{D})], with $d_i[n]$ denoting a random sampling binary coefficient, which is equal to 1 if $i\in \S[n]$, and 0 otherwise (i.e., $\S[n]$ represents the \textit{instantaneous}, random sampling set at time $n$); and $\bv[n]\in\mathbb{C}^N$ is zero-mean, spatially and temporally independent observation noise, with covariance matrix $\mR_v=\diag\{r_1^2,\ldots,r_N^2\}$. The estimation task consists in recovering the vector $\bx$ (or, equivalently, its GFT $\bs_{\F}$) from the noisy, streaming, and partial observations $\by[n]$ in (\ref{lin_observation}). Following an LMS approach \cite{sayed2011adaptive}, from (\ref{lin_observation}), the optimal estimate for $\bs_{\F}$ can be found as the vector that solves the following optimization problem:
\begin{align}
\label{LMS_problem}
&\min_{\boldsymbol{s}} \; \mathbb{E}\, \|\mD_{\S[n]}(\by[n]-\mU_{\F}\bs)\|^2
\end{align}
where in (\ref{LMS_problem}) we have exploited the fact that $\mD_{\S[n]}$ is an idempotent matrix for any fixed $n$ [cf. (\ref{D})]. An LMS-type solution optimizes (\ref{LMS_problem}) by means of a stochastic steepest-descent procedure, relying only on instantaneous information. Thus, letting $\widehat{\bx}[n]$ be the current estimates of vector $\bx$, the LMS algorithm for graph signals evolves as illustrated in Algorithm 2 \cite{dilo2016adaptive}.

\begin{algorithm}[t]
\vspace{.1cm}
Start with random $\widehat{\bx}[0]$. Given a step-size $\mu>0$, for $n\geq0$, repeat:
\begin{equation}\label{LMS}
\widehat{\bx}[n+1]=\widehat{\bx}[n]+\mu\,\mB_{\F}\mD_{\S[n]}\left(\by[n]-\widehat{\bx}[n]\right).\nonumber \vspace{-.3cm}
\end{equation}
\protect\caption{\textbf{: LMS on Graphs}}
\end{algorithm}

In the sequel, we illustrate how the design of the sampling strategy affects the reconstruction capability of Algorithm 2. To this aim, let us denote the \textit{expected sampling set} by $\overline{\S}=\{i=1,\ldots,N \,|\, p_i>0\}$, i.e., the set of nodes of the graph that are sampled with a probability $p_i=\mathbb{E}\{d_i[n]\}$ strictly greater than zero. Also, let $\overline{\mathcal{S}}_c$ be the complement set of $\overline{\S}$. Then, the following results illustrates the conditions for adaptive recovery of graph signals \cite{dilo2016adaptive,di2017adaptive}.

\begin{theorem}
Any $\F$-bandlimited graph signal can be reconstructed via the adaptive Algorithm 2 if and only if
\begin{equation}\label{|EDcB|<1}
\left\| \mD_{\,\overline{\mathcal{S}}_c}\mU_{\F}\right\|_2 < 1,
\end{equation}
i.e., if there are no $\F$-bandlimited signals that are perfectly localized on $\overline{\mathcal{S}_c}$.
\end{theorem}

\noindent Differently from batch sampling and recovery of graph signals, see, e.g.,
\cite{pesenson2008sampling,narang2013signal,chen2015discrete,tsitsvero2015signals,marques2016sampling}, condition (\ref{|EDcB|<1}) depends on the \textit{expected} sampling set. In particular, it implies that there are no $\F$-bandlimited signals that are perfectly localized over the set $\oline{\S}_c$. As a consequence, the adaptive Algorithm 2 with probabilistic sampling does not need to collect all the data necessary to reconstruct one-shot the graph signal at each iteration, but can learn acquiring the needed information over time. The only important thing required by condition (\ref{|EDcB|<1}) is that a sufficiently large number of nodes is sampled in \textit{expectation}.

We now illustrate the mean-square performance of Algorithm 2. The main results are summarized in the following Theorem \cite{di2017adaptive}.

\begin{theorem}
Assume spatial and temporal independence of the random variables extracted by the sampling process $\{d_i[n]\}_{i,n}$. Then, for any initial condition, Algorithm 2 is mean-square error stable if the sampling probability vector $\bp$ and the step-size $\mu$ satisfy (\ref{|EDcB|<1}) and
$$\displaystyle 0< \mu < \frac{2\lambda_{\min}\left(\mU_{\F}^H\,{\rm diag}(\bp)\mU_{\F}\right)}{\lambda^2_{\max}\left(\mU_{\F}^H\,{\rm diag}(\bp)\mU_{\F}\right)}.$$
Furthermore, under a small step-size assumtpion, the MSE writes as:
\begin{align}\label{MSD}
{\rm MSE}&=\lim_{n\rightarrow\infty}\,\mathbb{E}\|\widehat{\bx}[n]-\bx[n]\|^2 \nonumber\\
&=\,\frac{\mu}{2}\,{\rm Tr}\left[\left(\mU_{\F}^H\,{\rm diag}(\bp)\mU_{\F}\right)^{-1}\mU_{\F}^H\,{\rm diag}(\bp)\mR_v\mU_{\F}\right]
\end{align}
and the convergence rate $\alpha$ is well approximated by
\begin{equation}\label{learning_rate}
\alpha=
 1-2\mu \lambda_{\min}\left(\mU_{\F}^H\,{\rm diag}(\bp)\mU_{\F}\right).
\end{equation}
\end{theorem}

The results of Theorem 5 are instrumental to devise optimal probabilistic sampling strategies for Algorithm 2, which are described in the sequel.


\subsection{Probabilistic Sampling Strategies}

We consider a sampling design that seeks for the probability vector $\bp$ that minimizes the total sampling rate over the graph, i.e., $\mathbf{1}^T\bp$, while guaranteeing a target performance in terms of MSE in (\ref{MSD}) and of convergence rate in (\ref{learning_rate}) \cite{di2017adaptive}.
Then, the optimization problem can be cast as:
\begin{equation}\label{min_sampling_rate_problem}
\setlength{\jot}{6pt}
\begin{aligned}
&\hspace{-.3cm}\min_{\boldsymbol{p}} \;\; \mathbf{1}^T\bp  \\
&\hbox{s.t.} \;\;\;\,\lambda_{\min}\left(\mU_{\F}^H\,{\rm diag}(\bp)\mU_{\F}\right)\geq \displaystyle\frac{1-\bar{\alpha}}{2\mu}, \\
&\;\;\;\;\;\;\;\;{\rm Tr}\left[\left(\mU_{\F}^H\,{\rm diag}(\bp)\mU_{\F}\right)^{-1}\mU_{\F}^H\,{\rm diag}(\bp)\mR_v\mU_{\F}\right]\leq\frac{2\gamma}{\mu},\\
&\;\;\;\;\;\;\;\;\mathbf{0}\leq \bp \leq \bp^{\max}.
\end{aligned}
\end{equation}
The first constraint imposes that the convergence rate of the algorithm is larger than a desired value, i.e., $\alpha$ in (\ref{learning_rate}) is smaller than a target value, say, e.g., $\bar{\alpha}\in(0,1)$. Furthermore, as illustrated in \cite{di2017adaptive}, the first constraint on the convergence rate also guarantees adaptive signal reconstruction, i.e., condition (\ref{|EDcB|<1}) holds true. The second constraint guarantees a target mean-square performance, i.e., the MSE in (\ref{MSD}) must be less than or equal to a prescribed value, say, e.g., $\gamma>0$. Finally, the last constraint limits the probability vector to lie in the box $p_i\in[0,p^{\max}_i]$, for all $i$, with $0\leq p^{\max}_i\leq 1$ denoting an upper bound on the sampling probability at each node that might depend on external factors such as, e.g., limited energy, processing, and/or communication resources, node or communication failures, etc.

Unfortunately, problem (\ref{min_sampling_rate_problem}) is non-convex, due to the presence of the non-convex constraint on the MSE. To handle the non-convexity of (\ref{min_sampling_rate_problem}), we exploit an upper bound of the MSE function in (\ref{MSD}), given by:
\begin{equation}\label{MSD_bound}
{\rm MSE}(\bp)\,\leq\,\overline{{\rm MSE}}(\bp)\triangleq \frac{\mu}{2}\,\frac{{\rm Tr}\left(\mU_{\F}^H\,{\rm diag}(\bp)\mR_v\mU_{\F}\right)}{\lambda_{\min}\left(\mU_{\F}^H\,{\rm diag}(\bp)\mU_{\F}\right)}, \quad \hbox{for all $\bp\in \mathbb{R}^N$.}
\end{equation}
Of course, replacing the MSE function with the bound (\ref{MSD_bound}), the second constraint in (\ref{min_sampling_rate_problem}) is always satisfied. Furthermore, the function in (\ref{MSD_bound}) has the nice property to be pseudo-convex, since it is the ratio between a convex and a concave function, which are both differentiable and positive for all $\bp$ satisfying the other constraints \cite{avriel2010generalized}. Thus, exploiting the upper bound (\ref{MSD_bound}), we can formulate a surrogate optimization problem for the selection of the probability vector $\bp$, which can be cast as:
\begin{equation}\label{min_sampling_rate_problem2}
\setlength{\jot}{6pt}
\begin{aligned}
&\hspace{-.8cm}\min_{\boldsymbol{p}} \;\; \mathbf{1}^T\bp  \\
&\hspace{-.4cm}\hbox{subject to} \\
&\;\lambda_{\min}\left(\mU_{\F}^H\,{\rm diag}(\bp)\mU_{\F}\right)\geq \displaystyle\frac{1-\bar{\alpha}}{2\mu}, \\
&\;\frac{{\rm Tr}\left(\mU_{\F}^H\,{\rm diag}(\bp)\mC_v\mU_{\F}\right)}{\lambda_{\min}\left(\mU_{\F}^H\,{\rm diag}(\bp)\mU_{\F}\right)}\leq\frac{2\gamma}{\mu},\\
&\;\mathbf{0}\leq \bp \leq \bp^{\max}.
\end{aligned}
\end{equation}
Since the sublevel sets of pseudo-convex functions are convex sets \cite{avriel2010generalized}, it is straightforward to see that the approximated problem (\ref{min_sampling_rate_problem2}) is convex, and its global solution can be found using efficient numerical tools \cite{boyd2004convex}.

\subsection{Distributed Adaptive Recovery}

The implementation of Algorithm 2 would require to collect all the data $\{y_i[n]\}_{i:d_i[n]=1}$, for all $n$, in a single processing unit that performs the computation. In many practical systems, data are collected in a distributed network, and sharing local information with a central processor might be either unfeasible or not efficient, owing to the large volume of data, time-varying network topology, and/or privacy issues. Motivated by these observations, in this section we extend the LMS strategy in Algorithm 2 to a distributed setting, where the nodes perform the reconstruction task via online in-network processing, only exchanging data between neighbors defined over a sparse (but connected) communication network, which is described by the graph ${\cal G}_c=(\mathcal{V},\mathcal{E}_c)$. Proceeding as in \cite{di2016distAdaGraph} to derive distributed solution methods for problem (\ref{LMS_problem}), let us introduce local copies $\{\bs_i\}_{i=1}^N$ of the global variable $\bs$, and recast problem (\ref{LMS_problem}) in the following equivalent form:
\begin{align}\label{diffusion_LMS_problem}
&\min_{\{\boldsymbol{s}_i\}_{i=1}^N} \;\; \sum_{i=1}^N \;\mathbb{E} \left|d_i[n]\left(y_i[n]-\bu_{\F,i}^H\bs_i\right)\right|^2\\
&\qquad \hbox{subject to} \quad \bs_i=\bs_j \quad \hbox{for all $i=1,\ldots,N, \;\; j\in\mathcal{N}_i$,} \nonumber
\end{align}
where $\bu_{\F,i}^H$ is the $i$-th row of matrix $\mU_{\F}$ (supposed to be known at node $i$, or computable in distributed fashion, see, e.g., \cite{di2014distributed}), and $\mathcal{N}_i=\{j|a_{ij}>0\}$ is the local neighborhood of node $i$. To solve problem (\ref{LMS_problem}), we consider an Adapt-Then-Combine (ATC) diffusion strategy \cite{di2016distAdaGraph}, and the resulting algorithm is reported in Algorithm 3. The first step in (\ref{ATC diffusion}) is an adaptation step, where the intermediate estimate $\boldsymbol{\psi}_{i}[n]$ is updated adopting the current observation taken by node $i$, i.e. $y_i[n]$.  The second step is a diffusion step where the intermediate estimates $\boldsymbol{\psi}_{j}[n]$, from the (extended) spatial neighborhood $\overline{\mathcal{N}}_i=\mathcal{N}_i\bigcup \{i\}$, are combined through the weighting coefficients $\{w_{ij}\}$. Several possible combination rules have been proposed in the literature, such as the Laplacian or the Metropolis-Hastings weights, see, e.g. \cite{Barb-Sard-Dilo}, \cite{xiao2007distributed}, \cite{Cattivelli-Sayed}. Finally, given the estimate $\bs_{i}[n]$ of the GFT at node $i$ and time $n$, the last step produces the estimate $x_i[n+1]$ of the graph signal value at node $i$ [cf. (\ref{lin_observation})].
\begin{algorithm}[t]
\vspace{.1cm}
Start with random $\bs_{i}[0]$, for all $i\in \V$. Given combination weights $\{w_{ij}\}_{i,j}$, step-sizes $\mu_i>0$, for each time $n\geq0$ and for each node $i$, repeat:
\begin{align}\label{ATC diffusion}
&\boldsymbol{\psi}_{i}[n]=\bs_{i}[n]+\mu_i d_i[n] \bu_{\F,i}(y_i[n]-\bu_{\F,i}^H\bs_i[n]) \hspace{.4cm} \hbox{(adaptation)}\nonumber\\
&\bs_i[n+1]=\sum_{j\in \overline{\mathcal{N}}_i} w_{ij} \boldsymbol{\psi}_j[n] \hspace{3.6cm} \hbox{(diffusion)}\\
&x_i[n+1]=\bu_{\F,i}^H\bs_i[n+1]  \hspace{3.2cm} \hbox{(reconstruction)}\nonumber
\end{align}
\caption{\textbf{: Diffusion LMS on Graphs}}
\end{algorithm}
Here, we assume that graphs ${\cal G}$ (i.e., the one used for GSP) and ${\cal G}_c$ (i.e., the one describing the communication pattern among nodes) might have in general distinct topologies. We remark that both graphs play an important role in the proposed distributed processing strategy (\ref{ATC diffusion}). First, the processing graph determines the structure of the regression data $\bu_{\F,i}$ used in the adaptation step of (\ref{ATC diffusion}). In fact, $\{\bu_{\F,i}^H\}_i$ are the rows of the matrix $\mU_{\F}$, whose columns are the eigenvectors of the Laplacian matrix associated with the set of support frequencies $\F$. Then, the topology of the communication graph determines how information is spread all over the network through the diffusion step in (\ref{ATC diffusion}). This illustrates how, when reconstructing graph signals in a distributed manner, we have to take into account both the processing and communication aspects of the problem.

\subsection{Numerical Results}

In this section, we first illustrate the performance of the probabilistic sampling method in (\ref{min_sampling_rate_problem2}) over the IEEE 118 bus graph. Then, we consider an application to dynamic inference of brain activity.\\

\begin{figure}[t]
 \centering
 {\includegraphics[width=7cm]{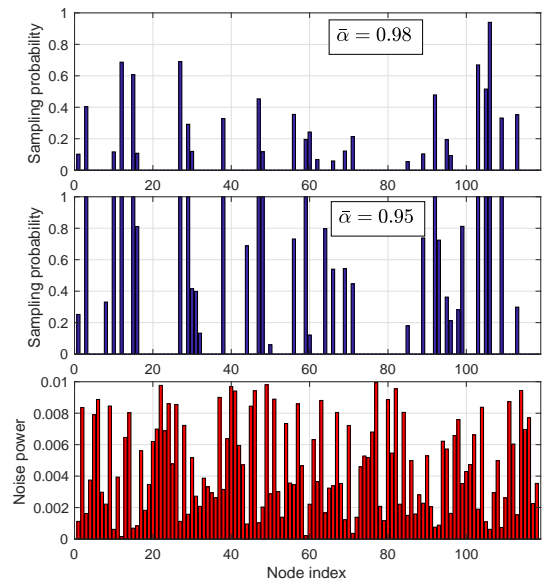}}
 \caption{Optimal probabilistic sampling over the IEEE118 graph.}\label{fig:Prob}
\end{figure}

\noindent \textbf{Optimal probabilistic sampling.} As a first example, let us consider an application to the IEEE 118 Bus Test Case in Fig. \ref{Nets}(a). The spectral content of the graph signal is assumed to be limited to the first ten eigenvectors of the Laplacian matrix of the graph. The observation noise in (\ref{lin_observation}) is zero-mean, Gaussian, with a diagonal covariance matrix $\mR_v$, where each element is illustrated in Fig. \ref{fig:Prob} (bottom). The other parameters are: $\mu=0.1$, and $\gamma=10^{-3}$. Then, in Fig. \ref{fig:Prob} (top and middle), we plot the optimal probability vector obtained solving (\ref{min_sampling_rate_problem2}), for two different values of $\bar{\alpha}$. In all cases, the constraints on the MSE and convergence rate are attained strictly. From Fig. \ref{fig:Prob} (top and middle), we notice how the method increases the sampling rate if we require a faster convergence (i.e., a smaller value of $\bar{\alpha}$); it also finds a very sparse probability vector and usually avoids to assign large sampling probabilities to nodes having large noise variances. Interestingly, with the proposed formulation, sparse sampling patterns are obtained thanks to the optimization of the sampling probabilities, which are already real numbers, without resorting to any relaxation of complex integer optimization problems [cf. (\ref{Conv_relax})].\\

\noindent \textbf{Inference of brain activity.} The last example  presents  test  results  on  Electrocorticography (ECoG) data,  captured through  experiments  conducted  in  an epilepsy study \cite{kramer2008emergent}. Data were collected over a  period  of  five days, where the electrodes recorded 76 ECoG time series, consisting of voltage levels measured in different regions of the brain. Two  temporal intervals of interest were picked for analysis, namely, the preictal and ictal intervals. In the sequel, we focus on the ictal interval. Further  details  about  data  acquisition and  pre-processing are provided in \cite{kramer2008emergent}. The GFT matrix $\mU_{\F}$ is learnt from the first 200 samples of ictal data, using the method proposed in \cite{gavish2017optimal}, and imposing a bandwidth equal to $|\F|=30$. In Fig. \ref{fig:track_brain}, we illustrate the true behavior of the ECoG present at an unobserved electrode chosen at random, over the first 400 samples of ictal data, along with estimate carried out using Algorithm 2 (with $\mu=1.5$). The sampling set is fixed over time (i.e., $p_i=1$ for all $i$), and chosen according to the E-optimal design in (\ref{D-design}), selecting $32$ samples. As we can notice from Fig. \ref{fig:track_brain}, the method is capable to efficiently infer and track the unknown dynamics of ECoG data at unobserved regions of the brain.

\begin{figure}[t]
\centering
\includegraphics[width=7cm,height=5cm]{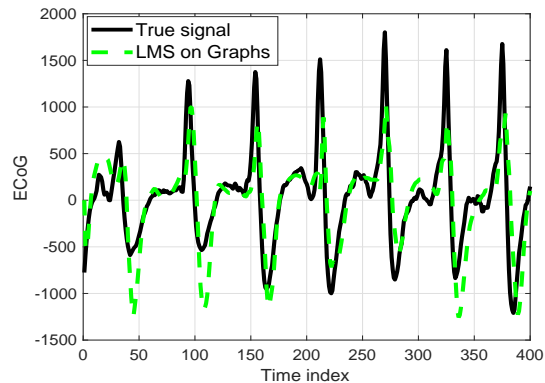}
\caption{True ECoG and estimate across time.}
\label{fig:track_brain}
\end{figure}

\section{Conclusions}

In this chapter, we have reviewed some of the methods recently proposed to sample and interpolate signals defined over graphs. First, we have recalled the conditions for perfect recovery under a bandlimited assumption. Second, we have illustrated sampling strategies, based on greedy methods or convex relaxations, aimed at reducing the effect of noise or aliasing on the recovered signal. Then, we considered $\ell_1$-norm reconstruction, which allows perfect recovery of graph signals in the presence of a strong impulsive noise over a limited number of nodes. Finally, adaptive methods based on (possibly distributed) LMS strategies were illustrated to enable tracking and recovery of time-varying signals over graphs. Several interesting problems need further investigation, e.g., sampling and recovery in the presence of directed/switching topologies, sampling adaptation in time-varying scenarios, distributed implementations, and the extension of GSP methods to incorporate multi-way relationships among data, e.g., under the form of hypergraphs or simplicial complexes.

\section*{REFERENCES}
\bibliographystyle{Vancouver-Numbered-Style_3_}
\bibliography{refs}

\begin{thebibliography*}{10}
\providecommand{\url}[1]{{\tt #1}}
\providecommand{\urlprefix}{URL: }
\expandafter\ifx\csname urlstyle\endcsname\relax
  \providecommand{\doi}[1]{doi:\discretionary{}{}{}#1}\else
  \providecommand{\doi}{doi:\discretionary{}{}{}\begingroup
  \urlstyle{rm}\Url}\fi
\providecommand{\bibinfo}[2]{#2}
\providecommand{\eprint}[2][]{\url{#2}}
\makeatletter\def\@biblabel#1{\bibinfo{label}{#1.}}\makeatother

\bibtype{Article}%
\bibitem{shuman2013emerging}
\bibinfo{author}{Shuman DI}, \bibinfo{author}{Narang SK},
  \bibinfo{author}{Frossard P}, \bibinfo{author}{Ortega A},
  \bibinfo{author}{Vandergheynst P}, \bibinfo{title}{The emerging field of
  signal processing on graphs: Extending high-dimensional data analysis to
  networks and other irregular domains}. \bibinfo{journal}{IEEE Signal Proc
  Mag} \bibinfo{year}{2013};
  \bibinfo{volume}{30}(\bibinfo{number}{3}):\bibinfo{pages}{83--98}.

\bibtype{Article}%
\bibitem{sandryhaila2013discrete}
\bibinfo{author}{Sandryhaila A}, \bibinfo{author}{Moura JMF},
  \bibinfo{title}{Discrete signal processing on graphs}. \bibinfo{journal}{IEEE
  Transacions on Signal Processing} \bibinfo{year}{2013};
  \bibinfo{volume}{61}:\bibinfo{pages}{1644--1656}.

\bibtype{Article}%
\bibitem{sandryhaila2014discrete}
\bibinfo{author}{Sandryhaila A}, \bibinfo{author}{Moura JM},
  \bibinfo{title}{Discrete signal processing on graphs: Frequency analysis}.
  \bibinfo{journal}{IEEE Transactions on Signal Processing}
  \bibinfo{year}{2014};
  \bibinfo{volume}{62}(\bibinfo{number}{12}):\bibinfo{pages}{3042--3054}.

\bibtype{Article}%
\bibitem{chen2015signal}
\bibinfo{author}{Chen S}, \bibinfo{author}{Sandryhaila A},
  \bibinfo{author}{Moura JM}, \bibinfo{author}{Kovacevic J},
  \bibinfo{title}{Signal recovery on graphs: Variation minimization}.
  \bibinfo{journal}{IEEE Transactions on Signal Processing}
  \bibinfo{year}{2015};
  \bibinfo{volume}{63}(\bibinfo{number}{17}):\bibinfo{pages}{4609--4624}.

\bibtype{inproceedings}%
\bibitem{zhu2012approximating}
\bibinfo{author}{Zhu X}, \bibinfo{author}{Rabbat M},
  \bibinfo{title}{Approximating signals supported on graphs}. In:
  \bibinfo{booktitle}{IEEE Int. Conf. on Acoustics, Speech and Signal
  Processing}, \bibinfo{year}{2012}, pp. \bibinfo{pages}{3921--3924}.

\bibtype{Article}%
\bibitem{chen2015discrete}
\bibinfo{author}{Chen S}, \bibinfo{author}{Varma R},
  \bibinfo{author}{Sandryhaila A}, \bibinfo{author}{Kova{\v c}evi{\'c} J},
  \bibinfo{title}{Discrete signal processing on graphs: Sampling theory}.
  \bibinfo{journal}{IEEE Transactions on Signal Processing}
  \bibinfo{year}{2015}; \bibinfo{volume}{63}:\bibinfo{pages}{6510--6523}.

\bibtype{inproceedings}%
\bibitem{gadde2014active}
\bibinfo{author}{Gadde A}, \bibinfo{author}{Anis A}, \bibinfo{author}{Ortega
  A}, \bibinfo{title}{Active semi-supervised learning using sampling theory for
  graph signals}. In: \bibinfo{booktitle}{Proceedings of the 20th ACM SIGKDD
  international conference on Knowledge discovery and data mining},
  \bibinfo{organization}{ACM}, \bibinfo{year}{2014}, pp.
  \bibinfo{pages}{492--501}.

\bibtype{Article}%
\bibitem{janssen2008spatial}
\bibinfo{author}{Janssen S}, \bibinfo{author}{Dumont G},
  \bibinfo{author}{Fierens F}, \bibinfo{author}{Mensink C},
  \bibinfo{title}{Spatial interpolation of air pollution measurements using
  corine land cover data}. \bibinfo{journal}{Atmospheric Env}
  \bibinfo{year}{2008};
  \bibinfo{volume}{42}(\bibinfo{number}{20}):\bibinfo{pages}{4884--4903}.

\bibtype{Article}%
\bibitem{candes2009exact}
\bibinfo{author}{Cand{\`e}s EJ}, \bibinfo{author}{Recht B},
  \bibinfo{title}{Exact matrix completion via convex optimization}.
  \bibinfo{journal}{Foundations of Computational mathematics}
  \bibinfo{year}{2009};
  \bibinfo{volume}{9}(\bibinfo{number}{6}):\bibinfo{pages}{717}.

\bibtype{Article}%
\bibitem{gomez2016netflix}
\bibinfo{author}{Gomez-Uribe CA}, \bibinfo{author}{Hunt N}, \bibinfo{title}{The
  netflix recommender system: Algorithms, business value, and innovation}.
  \bibinfo{journal}{ACM Trans on Management Inform Systems}
  \bibinfo{year}{2016};
  \bibinfo{volume}{6}(\bibinfo{number}{4}):\bibinfo{pages}{13}.

\bibtype{Article}%
\bibitem{yamanishi2004protein}
\bibinfo{author}{Yamanishi Y}, \bibinfo{author}{Vert JP},
  \bibinfo{author}{Kanehisa M}, \bibinfo{title}{Protein network inference from
  multiple genomic data: a supervised approach}.
  \bibinfo{journal}{Bioinformatics} \bibinfo{year}{2004};
  \bibinfo{volume}{20}(\bibinfo{number}{suppl\_1}):\bibinfo{pages}{i363--i370}.

\bibtype{Article}%
\bibitem{pesenson2008sampling}
\bibinfo{author}{Pesenson IZ}, \bibinfo{title}{Sampling in {Paley-Wiener}
  spaces on combinatorial graphs}. \bibinfo{journal}{Trans of the American
  Mathematical Society} \bibinfo{year}{2008};
  \bibinfo{volume}{360}(\bibinfo{number}{10}):\bibinfo{pages}{5603--5627}.

\bibtype{inproceedings}%
\bibitem{narang2013signal}
\bibinfo{author}{Narang S}, \bibinfo{author}{Gadde A}, \bibinfo{author}{Ortega
  A}, \bibinfo{title}{Signal processing techniques for interpolation in graph
  structured data}. In: \bibinfo{booktitle}{IEEE Conf. on Acous., Speech and
  Sig. Proc.}, \bibinfo{year}{2013}, pp. \bibinfo{pages}{5445--5449}.

\bibtype{inproceedings}%
\bibitem{anis2014towards}
\bibinfo{author}{Anis A}, \bibinfo{author}{Gadde A}, \bibinfo{author}{Ortega
  A}, \bibinfo{title}{Towards a sampling theorem for signals on arbitrary
  graphs}. In: \bibinfo{booktitle}{Acoustics, Speech and Signal Processing
  (ICASSP), 2014 IEEE International Conference on},
  \bibinfo{organization}{IEEE}, \bibinfo{year}{2014}, pp.
  \bibinfo{pages}{3864--3868}.

\bibtype{Article}%
\bibitem{wang2014local}
\bibinfo{author}{Wang X}, \bibinfo{author}{Liu P}, \bibinfo{author}{Gu Y},
  \bibinfo{title}{Local-set-based graph signal reconstruction}.
  \bibinfo{journal}{IEEE Trans on Signal Proc} \bibinfo{year}{2015};
  \bibinfo{volume}{63}(\bibinfo{number}{9}):\bibinfo{pages}{2432--2444}.

\bibtype{Article}%
\bibitem{tsitsvero2015signals}
\bibinfo{author}{Tsitsvero M}, \bibinfo{author}{Barbarossa S},
  \bibinfo{author}{Di~Lorenzo P}, \bibinfo{title}{Signals on graphs:
  Uncertainty principle and sampling}. \bibinfo{journal}{IEEE Transactions on
  Signal Processing} \bibinfo{year}{2016};
  \bibinfo{volume}{64}(\bibinfo{number}{18}):\bibinfo{pages}{4845--4860}.

\bibtype{Article}%
\bibitem{marques2016sampling}
\bibinfo{author}{Marques AG}, \bibinfo{author}{Segarra S},
  \bibinfo{author}{Leus G}, \bibinfo{author}{Ribeiro A},
  \bibinfo{title}{Sampling of graph signals with successive local
  aggregations}. \bibinfo{journal}{IEEE Transactions on Signal Processing}
  \bibinfo{year}{2016};
  \bibinfo{volume}{64}(\bibinfo{number}{7}):\bibinfo{pages}{1832--1843}.

\bibtype{Article}%
\bibitem{chamon2017greedy}
\bibinfo{author}{Chamon LF}, \bibinfo{author}{Ribeiro A},
  \bibinfo{title}{Greedy sampling of graph signals}.
  \bibinfo{journal}{arXiv:170401223} \bibinfo{year}{2017}; .

\bibtype{inproceedings}%
\bibitem{tremblay2016compressive}
\bibinfo{author}{Tremblay N}, \bibinfo{author}{Puy G},
  \bibinfo{author}{Gribonval R}, \bibinfo{author}{Vandergheynst P},
  \bibinfo{title}{Compressive spectral clustering}. In:
  \bibinfo{booktitle}{International Conference on Machine Learning},
  \bibinfo{year}{2016}, pp. \bibinfo{pages}{1002--1011}.

\bibtype{Article}%
\bibitem{anis2016efficient}
\bibinfo{author}{Anis A}, \bibinfo{author}{Gadde A}, \bibinfo{author}{Ortega
  A}, \bibinfo{title}{Efficient sampling set selection for bandlimited graph
  signals using graph spectral proxies}. \bibinfo{journal}{IEEE Trans on Sig
  Proc} \bibinfo{year}{2016};
  \bibinfo{volume}{64}(\bibinfo{number}{14}):\bibinfo{pages}{3775--3789}.

\bibtype{inproceedings}%
\bibitem{narang2010local}
\bibinfo{author}{Narang SK}, \bibinfo{author}{Ortega A}, \bibinfo{title}{Local
  two-channel critically sampled filter-banks on graphs}. In:
  \bibinfo{booktitle}{Image Proc., 17th IEEE Int. Conf. on},
  \bibinfo{organization}{IEEE}, \bibinfo{year}{2010}, pp.
  \bibinfo{pages}{333--336}.

\bibtype{Article}%
\bibitem{nguyen2015downsampling}
\bibinfo{author}{Nguyen HQ}, \bibinfo{author}{Do MN},
  \bibinfo{title}{Downsampling of signals on graphs via maximum spanning
  trees.} \bibinfo{journal}{IEEE Trans Signal Processing} \bibinfo{year}{2015};
  \bibinfo{volume}{63}(\bibinfo{number}{1}):\bibinfo{pages}{182--191}.

\bibtype{Article}%
\bibitem{chen2016signal}
\bibinfo{author}{Chen S}, \bibinfo{author}{Varma R}, \bibinfo{author}{Singh A},
  \bibinfo{author}{Kova{\v{c}}evi{\'c} J}, \bibinfo{title}{Signal recovery on
  graphs: Fundamental limits of sampling strategies}. \bibinfo{journal}{IEEE
  Trans on Sig and Inf Proc over Net} \bibinfo{year}{2016};
  \bibinfo{volume}{2}(\bibinfo{number}{4}):\bibinfo{pages}{539--554}.

\bibtype{Article}%
\bibitem{puy2016random}
\bibinfo{author}{Puy G}, \bibinfo{author}{Tremblay N},
  \bibinfo{author}{Gribonval R}, \bibinfo{author}{Vandergheynst P},
  \bibinfo{title}{Random sampling of bandlimited signals on graphs}.
  \bibinfo{journal}{Applied and Computational Harmonic Analysis}
  \bibinfo{year}{2016}; .

\bibtype{inproceedings}%
\bibitem{Tre17}
\bibinfo{author}{N~Tremblay P-O~Amblard SB}, \bibinfo{title}{Graph sampling
  with determinantal processes}. In: \bibinfo{booktitle}{European Signal
  Processing Conference (Eusipco)}, \bibinfo{year}{2017}.

\bibtype{Article}%
\bibitem{dilo2016adaptive}
\bibinfo{author}{Di~Lorenzo P}, \bibinfo{author}{Barbarossa S},
  \bibinfo{author}{Banelli P}, \bibinfo{author}{Sardellitti S},
  \bibinfo{title}{Adaptive least mean squares estimation of graph signals}.
  \bibinfo{journal}{IEEE Trans on Sig and Inf Proc over Net}
  \bibinfo{year}{2016};
  \bibinfo{volume}{2}(\bibinfo{number}{4}):\bibinfo{pages}{555--568}.

\bibtype{Article}%
\bibitem{di2016distAdaGraph}
\bibinfo{author}{Di~Lorenzo P}, \bibinfo{author}{Banelli P},
  \bibinfo{author}{Barbarossa S}, \bibinfo{author}{Sardellitti S},
  \bibinfo{title}{Distributed adaptive learning of graph signals}.
  \bibinfo{journal}{IEEE Transactions on Signal Processing}
  \bibinfo{year}{2017};
  \bibinfo{volume}{65}(\bibinfo{number}{16}):\bibinfo{pages}{4193--4208}.

\bibtype{Article}%
\bibitem{romero2016kernel}
\bibinfo{author}{Romero D}, \bibinfo{author}{Ioannidis VN},
  \bibinfo{author}{Giannakis GB}, \bibinfo{title}{Kernel-based reconstruction
  of space-time functions on dynamic graphs}. \bibinfo{journal}{IEEE J of Sel
  Topics in Sig Proc} \bibinfo{year}{2017};
  \bibinfo{volume}{11}(\bibinfo{number}{6}):\bibinfo{pages}{856--869}.

\bibtype{Article}%
\bibitem{wang2015distributed}
\bibinfo{author}{Wang X}, \bibinfo{author}{Wang M}, \bibinfo{author}{Gu Y},
  \bibinfo{title}{A distributed tracking algorithm for reconstruction of graph
  signals}. \bibinfo{journal}{IEEE Journal of Selected Topics in Signal
  Processing} \bibinfo{year}{2015};
  \bibinfo{volume}{9}(\bibinfo{number}{4}):\bibinfo{pages}{728--740}.

\bibtype{Book}%
\bibitem{godsil2013algebraic}
\bibinfo{author}{Godsil C}, \bibinfo{author}{Royle GF},
  \bibinfo{title}{Algebraic graph theory}, vol. \bibinfo{volume}{207}.
  \bibinfo{publisher}{Springer Science \& Business Media},
  \bibinfo{year}{2013}.

\bibtype{Article}%
\bibitem{eldar2003sampling}
\bibinfo{author}{Eldar YC}, \bibinfo{title}{Sampling with arbitrary sampling
  and reconstruction spaces and oblique dual frame vectors}.
  \bibinfo{journal}{Journal of Fourier Analysis and Applications}
  \bibinfo{year}{2003};
  \bibinfo{volume}{9}(\bibinfo{number}{1}):\bibinfo{pages}{77--96}.

\bibtype{Book}%
\bibitem{kay1993fundamentals}
\bibinfo{author}{Kay SM}, \bibinfo{title}{Fundamentals of statistical signal
  processing}. \bibinfo{publisher}{Prentice Hall PTR}, \bibinfo{year}{1993}.

\bibtype{Book}%
\bibitem{winer1971statistical}
\bibinfo{author}{Winer BJ}, \bibinfo{author}{Brown DR},
  \bibinfo{author}{Michels KM}, \bibinfo{title}{Statistical principles in
  experimental design}, vol.~\bibinfo{volume}{2}.
  \bibinfo{publisher}{McGraw-Hill New York}, \bibinfo{year}{1971}.

\bibtype{Article}%
\bibitem{nemhauser1988integer}
\bibinfo{author}{Nemhauser GL}, \bibinfo{author}{Wolsey LA},
  \bibinfo{title}{Integer and combinatorial optimization. interscience series
  in discrete mathematics and optimization}. \bibinfo{journal}{ed: John Wiley
  \& Sons} \bibinfo{year}{1988}; .

\bibtype{incollection}%
\bibitem{lovasz1983submodular}
\bibinfo{author}{Lov{\'a}sz L}, \bibinfo{title}{Submodular functions and
  convexity}. In: \bibinfo{booktitle}{Mathematical Programming The State of the
  Art}, \bibinfo{publisher}{Springer}, \bibinfo{year}{1983}; pp.
  \bibinfo{pages}{235--257}.

\bibtype{Article}%
\bibitem{nemhauser1978analysis}
\bibinfo{author}{Nemhauser GL}, \bibinfo{author}{Wolsey LA},
  \bibinfo{author}{Fisher ML}, \bibinfo{title}{An analysis of approximations
  for maximizing submodular set functions-i}. \bibinfo{journal}{Mathematical
  Programming} \bibinfo{year}{1978};
  \bibinfo{volume}{14}(\bibinfo{number}{1}):\bibinfo{pages}{265--294}.

\bibtype{Article}%
\bibitem{summerscorrection}
\bibinfo{author}{Summers TH}, \bibinfo{author}{Cortesi FL},
  \bibinfo{author}{Lygeros J}, \bibinfo{title}{Correction to “on
  submodularity and controllability in complex dynamical networks”}
  \bibinfo{year}{Available at:
  http://wwwutdallasedu/~tylersummers/papers/TCNS$\_$Correctionpdf}; .

\bibtype{inproceedings}%
\bibitem{shamaiah2010greedy}
\bibinfo{author}{Shamaiah M}, \bibinfo{author}{Banerjee S},
  \bibinfo{author}{Vikalo H}, \bibinfo{title}{Greedy sensor selection:
  Leveraging submodularity}. In: \bibinfo{booktitle}{Decision and Control, 2010
  49th IEEE Conf. on}, \bibinfo{organization}{IEEE}, \bibinfo{year}{2010}, pp.
  \bibinfo{pages}{2572--2577}.

\bibtype{Article}%
\bibitem{joshi2009sensor}
\bibinfo{author}{Joshi S}, \bibinfo{author}{Boyd S}, \bibinfo{title}{Sensor
  selection via convex optimization}. \bibinfo{journal}{IEEE Transactions on
  Signal Processing} \bibinfo{year}{2009};
  \bibinfo{volume}{57}(\bibinfo{number}{2}):\bibinfo{pages}{451--462}.

\bibtype{Article}%
\bibitem{chepuri2015sparsity}
\bibinfo{author}{Chepuri SP}, \bibinfo{author}{Leus G},
  \bibinfo{title}{Sparsity-promoting sensor selection for non-linear
  measurement models}. \bibinfo{journal}{IEEE Transactions on Signal
  Processing} \bibinfo{year}{2015};
  \bibinfo{volume}{63}(\bibinfo{number}{3}):\bibinfo{pages}{684--698}.

\bibtype{Article}%
\bibitem{chepuri2016sparse}
\bibinfo{author}{Chepuri SP}, \bibinfo{author}{Leus G}, \bibinfo{title}{Sparse
  sensing for distributed detection}. \bibinfo{journal}{IEEE Transactions on
  Signal Processing} \bibinfo{year}{2016};
  \bibinfo{volume}{64}(\bibinfo{number}{6}):\bibinfo{pages}{1446--1460}.

\bibtype{Book}%
\bibitem{boyd2004convex}
\bibinfo{author}{Boyd S}, \bibinfo{author}{Vandenberghe L},
  \bibinfo{title}{Convex optimization}. \bibinfo{publisher}{Cambridge
  university press}, \bibinfo{year}{2004}.

\bibtype{Article}%
\bibitem{sun2005simulation}
\bibinfo{author}{Sun K}, \bibinfo{author}{Zheng DZ}, \bibinfo{author}{Lu Q},
  \bibinfo{title}{A simulation study of obdd-based proper splitting strategies
  for power systems under consideration of transient stability}.
  \bibinfo{journal}{IEEE Transactions on Power Systems} \bibinfo{year}{2005};
  \bibinfo{volume}{20}(\bibinfo{number}{1}):\bibinfo{pages}{389--399}.

\bibtype{Article}%
\bibitem{pasqualetti2014controllability}
\bibinfo{author}{Pasqualetti F}, \bibinfo{author}{Zampieri S},
  \bibinfo{author}{Bullo F}, \bibinfo{title}{Controllability metrics,
  limitations and algorithms for complex networks}. \bibinfo{journal}{IEEE
  Trans on Control of Network Systems} \bibinfo{year}{2014};
  \bibinfo{volume}{1}(\bibinfo{number}{1}):\bibinfo{pages}{40--52}.

\bibtype{inproceedings}%
\bibitem{behrisch2011sumo}
\bibinfo{author}{Behrisch M}, \bibinfo{author}{Bieker L},
  \bibinfo{author}{Erdmann J}, \bibinfo{author}{Krajzewicz D},
  \bibinfo{title}{Sumo--simulation of urban mobility: an overview}. In:
  \bibinfo{booktitle}{Int. Conference on Adv. in System Simulation},
  \bibinfo{organization}{ThinkMind}, \bibinfo{year}{2011}.

\bibtype{Book}%
\bibitem{sayed2011adaptive}
\bibinfo{author}{Sayed AH}, \bibinfo{title}{Adaptive filters}.
  \bibinfo{publisher}{John Wiley \& Sons}, \bibinfo{year}{2011}.

\bibtype{Article}%
\bibitem{di2017adaptive}
\bibinfo{author}{Di~Lorenzo P}, \bibinfo{author}{Banelli P},
  \bibinfo{author}{Isufi E}, \bibinfo{author}{Barbarossa S},
  \bibinfo{author}{Leus G}, \bibinfo{title}{Adaptive graph signal processing:
  Algorithms and optimal sampling strategies}.
  \bibinfo{journal}{arXiv:170903726} \bibinfo{year}{2017}; .

\bibtype{Book}%
\bibitem{avriel2010generalized}
\bibinfo{author}{Avriel M}, \bibinfo{author}{Diewert WE},
  \bibinfo{author}{Schaible S}, \bibinfo{author}{Zang I},
  \bibinfo{title}{Generalized concavity}. \bibinfo{publisher}{SIAM},
  \bibinfo{year}{2010}.

\bibtype{Article}%
\bibitem{di2014distributed}
\bibinfo{author}{Di~Lorenzo P}, \bibinfo{author}{Barbarossa S},
  \bibinfo{title}{Distributed estimation and control of algebraic connectivity
  over random graphs}. \bibinfo{journal}{IEEE Transactions on Signal
  Processing} \bibinfo{year}{2014};
  \bibinfo{volume}{62}(\bibinfo{number}{21}):\bibinfo{pages}{5615--5628}.

\bibtype{inbook}%
\bibitem{Barb-Sard-Dilo}
\bibinfo{author}{Barbarossa S}, \bibinfo{author}{Sardellitti S},
  \bibinfo{author}{{Di Lorenzo} P}, \bibinfo{title}{Distributed Detection and
  Estimation in Wireless Sensor Networks}, vol.~\bibinfo{volume}{2}.
  \bibinfo{publisher}{Academic Press Library in Signal Processing},
  \bibinfo{year}{2014}; pp. \bibinfo{pages}{329--408}.

\bibtype{Article}%
\bibitem{xiao2007distributed}
\bibinfo{author}{Xiao L}, \bibinfo{author}{Boyd S}, \bibinfo{author}{Kim SJ},
  \bibinfo{title}{Distributed average consensus with least-mean-square
  deviation}. \bibinfo{journal}{Journal of Parallel and Distributed Computing}
  \bibinfo{year}{2007};
  \bibinfo{volume}{67}(\bibinfo{number}{1}):\bibinfo{pages}{33--46}.

\bibtype{Article}%
\bibitem{Cattivelli-Sayed}
\bibinfo{author}{Cattivelli FS}, \bibinfo{author}{Sayed AH},
  \bibinfo{title}{Diffusion {LMS} strategies for distributed estimation}.
  \bibinfo{journal}{IEEE Trans on Signal Processing} \bibinfo{year}{2010};
  \bibinfo{volume}{58}:\bibinfo{pages}{1035--1048}.

\bibtype{Article}%
\bibitem{kramer2008emergent}
\bibinfo{author}{Kramer MA}, \bibinfo{author}{Kolaczyk ED},
  \bibinfo{author}{Kirsch HE}, \bibinfo{title}{Emergent network topology at
  seizure onset in humans}. \bibinfo{journal}{Epilepsy research}
  \bibinfo{year}{2008};
  \bibinfo{volume}{79}(\bibinfo{number}{2}):\bibinfo{pages}{173--186}.

\bibtype{Article}%
\bibitem{gavish2017optimal}
\bibinfo{author}{Gavish M}, \bibinfo{author}{Donoho DL},
  \bibinfo{title}{Optimal shrinkage of singular values}. \bibinfo{journal}{IEEE
  Trans on Inf Theory} \bibinfo{year}{2017};
  \bibinfo{volume}{63}(\bibinfo{number}{4}):\bibinfo{pages}{2137--2152}.

\end{thebibliography*}

\end{document}